\documentclass[11pt]{article}
\usepackage{amsfonts,amsmath,amsthm,amssymb,dsfont, etex}

\usepackage{tikz}
\usetikzlibrary{arrows,automata}
\usepackage{comment, times}
\usepackage[margin=1in]{geometry}
\usepackage{soul,color}
\usepackage{graphicx,float,wrapfig}
\usepackage{mathrsfs}
\usepackage[usenames,dvipsnames]{pstricks}
\usepackage{epsfig}
\usepackage{pst-grad} 
\usepackage{pst-plot} 
\usepackage{mathdots}
\usepackage[linesnumbered,boxed,ruled,vlined]{algorithm2e}

\usepackage{url}

\usepackage{tabu}

\usepackage[normalem]{ulem}

\newtheorem{theo}{Theorem}[section]

\newtheorem{lemma}[theo]{Lemma}
\newtheorem{conj}[theo]{Conjecture}

\newtheorem{cor}[theo]{Corollary}
\theoremstyle{definition}
\newtheorem{defi}[theo]{Definition}
\theoremstyle{plain}
\newtheorem{rem}[theo]{Remark}

\newtheorem{ques}{Question}

\newenvironment{proofof}[1]{\begin{proof}[Proof of #1]}{\end{proof}}
\newenvironment{proofsketch}{\begin{proof}[Proof Sketch]}{\end{proof}}

\newenvironment{reminder}[1]{\bigskip
	\noindent {\bf Reminder of #1  }\em}{\smallskip}

\newcommand{\R}{\mathbb{R}}
\newcommand{\Ex}{\mathbb{E}}
\everymath{\displaystyle}

\newcommand{\alg}{\mathbb{A}}

\newcommand{\poly}{\operatorname*{poly}}

\newcommand{\BQP}{\mathsf{BQP}}

\newcommand{\SBP}{\mathsf{SBP}}

\newcommand{\MA}{\mathsf{MA}}

\newcommand{\PTIME}{\mathsf{P}}

\newcommand{\NP}{\mathsf{NP}}

\newcommand{\UPP}{\mathsf{UPP}}

\newcommand{\eps}{\epsilon}

\newcommand{\polylog}{\operatorname*{polylog}}

\newcommand{\OR}{\mathsf{OR}}

\renewcommand{\epsilon}{\varepsilon}

\def\ShowAuthNotes{1}
\ifnum\ShowAuthNotes=1
\newcommand{\authnote}[2]{\ \\ \textcolor{red}{\parbox{0.9\linewidth}{[{\footnotesize {\bf #1:} { {#2}}}]}}\newline}
\else
\newcommand{\authnote}[2]{}
\fi

\let\svfootnoterule\footnoterule
\renewcommand\footnoterule{\vfill\svfootnoterule}
\newcommand{\SAT}{\mathsf{SAT}}
\newcommand{\SETH}{\mathsf{SETH}}

\newcommand{\MaxIP}{\textsf{Max-IP}}
\newcommand{\IntMaxIP}{\textsf{$\mathbb{Z}$-Max-IP}}
\newcommand{\AllPairMaxIP}{\textsf{All-Pair-Max-IP}}
\newcommand{\OV}{\textsf{OV}}
\newcommand{\IntOV}{\textsf{$\mathbb{Z}$-OV}}
\newcommand{\Hopcroft}{\IntOV}

\newcommand{\OPT}{\textsf{OPT}}
\newcommand{\WOPT}{\widetilde{\textsf{OPT}}}

\newcommand{\DISJ}{\textsf{DISJ}}

\newcommand{\logstar}{\log^{*}}
\newcommand{\WT}{\widetilde}

\newcommand{\Mapprox}{multiplicative-approximating}
\newcommand{\Aapprox}{additive-approximating}

\newcommand{\LCSP}{\textsf{LCS-Closest-Pair}}

\newcommand{\promiseBPP}{\textsf{promiseBPP}}

\newcommand{\CRR}{\mathsf{CRR}}
\newcommand{\bm}{b_{\mathsf{micro}}}

\newcommand{\posR}{\mathbb{R}^{+}}
\newcommand{\posMaxIP}{\textsf{$\posR$-Max-IP}}
\newcommand{\pnMaxIP}{\{-1,1\}\text{-}\MaxIP}

\title{On The Hardness of Approximate and Exact (Bichromatic) Maximum Inner Product}

\author{Lijie Chen\thanks{Email: lijieche@mit.edu. Supported by an Akamai  Fellowship.}\\MIT}

\date{}

\begin{document}
	\maketitle
	
	\begin{abstract}
		In this paper we study the (Bichromatic) Maximum Inner Product Problem (\MaxIP), in which we are given sets $A$ and $B$ of vectors, and the goal is to find $a \in A$ and $b \in B$ maximizing inner product $a \cdot b$. $\MaxIP$ is very basic and serves as the base problem in the recent breakthrough of [Abboud et al., FOCS 2017] on hardness of approximation for polynomial-time problems. It is also used (implicitly) in the argument for hardness of exact $\ell_2$-Furthest Pair (and other important problems in computational geometry) in poly-log-log dimensions in [Williams, SODA 2018]. We have three main results regarding this problem.
		
		\begin{itemize}			
			\item \textbf{Characterization of Multiplicative Approximation}. First, we study the best multiplicative approximation ratio for Boolean $\MaxIP$ in sub-quadratic time. We show that, for $\MaxIP$ with two sets of $n$ vectors from $\{0,1\}^{d}$, there is an $n^{2 - \Omega(1)}$ time $\left( d/\log n \right)^{\Omega(1)}$-\Mapprox\ algorithm, and we show this is conditionally optimal, as such a $\left(d/\log n\right)^{o(1)}$-approximating algorithm would refute SETH.
			
			\item \textbf{Characterization of Additive Approximation}.Second, we achieve a similar characterization for the best additive approximation error to Boolean $\MaxIP$. We show that, for $\MaxIP$ with two sets of $n$ vectors from $\{0,1\}^{d}$, there is an $n^{2 - \Omega(1)}$ time $\Omega(d)$-\Aapprox\ algorithm, and this is conditionally optimal, as such an $o(d)$-approximating algorithm would refute SETH~[Rubinstein, STOC 2018].
			
			\item \textbf{$2^{O(\logstar n)}$-dimensional Hardness for Exact $\MaxIP$ Over The Integers.} Last, we revisit the hardness of solving $\MaxIP$ exactly for vectors with integer entries. We show that, under SETH, for $\MaxIP$ with sets of $n$ vectors from $\mathbb{Z}^{d}$ for some $d = 2^{O(\logstar n)}$, every exact algorithm requires $n^{2 - o(1)}$ time. With the reduction from [Williams, SODA 2018], it follows that $\ell_2$-Furthest Pair and Bichromatic $\ell_2$-Closest Pair in $2^{O(\logstar n)}$ dimensions require $n^{2 - o(1)}$ time.
			
		\end{itemize}
		The lower bound in our first result is a direct corollary of the new $\MA$ protocol for Set-Disjointness introduced in [Rubinstein, STOC 2018]. Our algorithms utilize the polynomial method and simple random sampling. Our second result follows from a new dimensionality self reduction from the Orthogonal Vectors problem for $n$ vectors from $\{0,1\}^{d}$ to $n$ vectors from $\mathbb{Z}^{\ell}$ using \emph{Chinese Remainder Theorem}, where $\ell = 2^{O(\logstar d)}$, dramatically improving the previous reduction in [Williams, SODA 2018].

	We also establish a connection between conditional lower bounds for exact $\MaxIP$ with integer entries and $\NP \cdot \UPP$ communication protocols for Set-Disjointness, parallel to the connection between conditional lower bounds for approximating $\MaxIP$ and $\MA$ communication protocols for Set-Disjointness. Moreover, as a side product, we obtain an $\MA$ communication protocol for Set-Disjointness with complexity $O\left(\sqrt{n\log n \log\log n}\right)$, slightly improving the $O\left(\sqrt{n} \log n\right)$ bound~[Aaronson and Wigderson, TOCT 2009], and approaching the $\Omega(\sqrt{n})$ lower bound~[Klauck, CCC 2003].
	\end{abstract}
	
	\section{Introduction}

We study the following fundamental problem from similarity search and statistics, which asks to find the most correlated pair in a dataset:

\begin{defi}[Bichromatic Maximum Inner Product ($\MaxIP$)]
	For $n,d \in \mathbb{N}$, the $\MaxIP_{n,d}$ problem is defined as: \emph{given two sets $A,B$ of vectors from $\{0,1\}^{d}$ compute}
	\[
	\OPT(A,B) := \max_{a \in A, b \in B} a \cdot b.
	\]
	
	We use $\IntMaxIP_{n,d}$ ($\mathbb{R}\text{-}\MaxIP_{n,d}$) to denote the same problem, but with $A,B$ being sets of vectors from $\mathbb{Z}^{d}$ ($\mathbb{R}^d$).
\end{defi}

\paragraph*{Hardness of Approximation $\MaxIP$.}
A natural brute-force algorithm solves $\MaxIP$ in $O(n^2 \cdot d)$-time. Assuming SETH\footnote{SETH (Strong Exponential Time Hypothesis) states that for every $\eps > 0$ there is a $k$ such that $k$-SAT cannot be solved in $O((2-\eps)^n)$ time~\cite{IP01-SETH}.}, there is no $n^{2 - \Omega(1)}$-time algorithm for $\MaxIP_{n, d}$ when $d = \omega(\log n)$~\cite{Wil05}. 

Despite being one of the most central problems in similarity search and having numerous applications~\cite{IM98,AI06,RR07,RG12,SL14,AINR14,AILRS15,AR15,NS15,SL15,Valiant15,AW15,KarKK16,ahle2016complexity,TG16,CP16,Chris17}, until recently it was unclear whether there could be a near-linear-time, $1.1$-approximating algorithm, before the recent breakthrough of Abboud, Rubinstein and Williams~\cite{ARW17-proceedings} (see~\cite{ARW17-proceedings} for a thorough discussion on the state of affairs on hardness of approximation in P before their work). 

In~\cite{ARW17-proceedings}, a framework for proving inapproximability results for problems in $\PTIME$ is established (the distributed PCP framework), from which it follows:

\begin{theo}[\cite{ARW17-proceedings}]\label{theo:ARW}
	Assuming SETH, there is no $2^{(\log n)^{1 - o(1)}}$-\Mapprox\ $n^{2 - \Omega(1)}$-time algorithm for $\MaxIP_{n, n^{o(1)}}$.
\end{theo}

Theorem~\ref{theo:ARW} is an exciting breakthrough for hardness of approximation in $\PTIME$, implying other important inapproximability results for a host of problems including Bichromatic LCS
Closest Pair Over Permutations, Approximate Regular Expression Matching, and Diameter in Product Metrics~\cite{ARW17-proceedings}. However, we still do not have a complete understanding of the approximation hardness of $\MaxIP$ yet. For instance, consider the following two concrete questions:

\begin{ques}~\label{ques:M-dont-know}
	Is there a $(\log n)$-\Mapprox\ $n^{2 - \Omega(1)}$-time algorithm for $\MaxIP_{n, \log^2 n}$? What about a $2$-\Mapprox\ for $\MaxIP_{n, \log^2 n}$?
\end{ques}

\begin{ques}~\label{ques:A-dont-know}
	Is there a $(d/\log n)$-\Aapprox\ $n^{2 - \Omega(1)}$-time algorithm for $\MaxIP_{n,d}$?
\end{ques}

We note that the lower bound from~\cite{ARW17-proceedings} cannot answer Question~1. Tracing the details of their proofs, one can see that it only shows approximation hardness for dimension $d = \log^{\omega(1)} n$. Question 2 concerning additive approximation is not addressed at all by~\cite{ARW17-proceedings}. Given the importance of $\MaxIP$, it is interesting to ask: 

\medskip

{\narrower

\emph{For what ratios $r$ do $n^{2-\Omega(1)}$-time $r$-approximation algorithms exist for $\MaxIP$?}

}

\medskip
Does the best-possible approximation ratio (in $n^{2 - \Omega(1)}$ time) relate to the dimensionality, in some way? 

In an important recent work, Rubinstein~\cite{Rubinstein2017closest} improved the distributed PCP construction in a very crucial way, from which one can derive more refined lower bounds on approximating $\MaxIP$. Building on its technique, in this paper we provide full \emph{characterizations}, determining essentially optimal multiplicative approximations and additive approximations to $\MaxIP$, under SETH. 

\paragraph*{Hardness of Exact $\IntMaxIP$.}
Recall that from~\cite{Wil05}, there is no $n^{2-\Omega(1)}$-time algorithm for exact Boolean $\MaxIP_{n,\omega(\log n)}$. Since in real life applications of similarity search, one often deals with real-valued data instead of just Boolean data, it is natural to ask about $\IntMaxIP$ (which is certainly a special case of $\mathbb{R}\text{-}\MaxIP$): what is the maximum $d$ such that $\IntMaxIP_{n,d}$ can be solved exactly in $n^{2 - \Omega(1)}$ time? 

Besides being interesting in its own right, there are also reductions from $\IntMaxIP$ to $\ell_2$-Furthest Pair and Bichromatic $\ell_2$-Closest Pair. Hence, lower bounds for $\IntMaxIP$ imply lower bounds for these two famous problems in computational geometry (see~\cite{Wil18} for a discussion on this topic).

Prior to our work, it was implicitly shown in~\cite{Wil18} that:

\begin{theo}[\cite{Wil18}]
	Assuming SETH, there is no $n^{2 - \Omega(1)}$-time algorithm for $\IntMaxIP_{n,\omega((\log\log n)^2)}$ with vectors of $O(\log n)$-bit entries.
\end{theo}

However, the best known algorithm for $\IntMaxIP$ runs in $n^{2 - \Theta(1/d)}$ time~\cite{matouvsek1992efficient,agarwal1991euclidean,yao1982constructing}\footnote{\cite{agarwal1991euclidean,yao1982constructing} are for $\ell_2$-Furthest Pair or Bichromatic $\ell_2$-Closest Pair. They also work for $\IntMaxIP$ as there are reductions from $\IntMaxIP$ to these two problems, see~\cite{Wil18} or Lemma~\ref{lm:Max-IP-to-furtherest-pair} and Lemma~\ref{lm:Max-IP-to-bichromatic-closest-pair}.}, hence there is still a gap between the lower bound and the best known upper bounds. To confirm these algorithms are in fact optimal, we would like to prove a lower bound with $\omega(1)$ dimensions.

In this paper, we significantly strength the previous lower bound from $\omega((\log\log n)^2)$ dimensions to $2^{O(\logstar n)}$ dimensions ($2^{O(\logstar n)}$ is an \emph{extremely slow-growing} function, see preliminaries for its formal definition).

\subsection{Our Results}

We use $\OV_{n,d}$ to denote the Orthogonal Vectors problem: given two sets of vectors $A,B$ each consisting of $n$ vectors from $\{0,1\}^d$, determine whether there are $a \in A$ and $b \in B$ such that $a \cdot b = 0$.\footnote{Here we use the bichromatic version of $\OV$ instead of the monochromatic one for convenience, as they are equivalent.} Similarly, we use $\IntOV_{n,d}$ to denote the same problem except for that $A,B$ consists of vectors from $\mathbb{Z}^d$ (which is also called Hopcroft's problem).

All our results are based on the following widely used conjecture about $\OV$:

\begin{conj}[Orthogonal Vectors Conjecture~(OVC)~\cite{Wil05,AVW14}]~\label{conj:OVC}
	For every $\eps > 0$, there exists a $c \ge 1$ such that $\OV_{n,d}$ requires $n^{2 - \eps}$ time when $d = c\log n$.
\end{conj}

OVC is a plausible conjecture as it is implied by the popular Strong Exponential Time Hypothesis~\cite{IP01-SETH,calabro2009complexity} on the time complexity of solving $k$-$\SAT$~\cite{Wil05,williams2014finding}.

\subsection*{Characterizations of Hardness of Approximate $\MaxIP$}

The first main result of our paper characterizes when there is a truly sub-quadratic time ($n^{2-\Omega(1)}$ time, for some universal constant hidden in the big-$\Omega$) $t$-\Mapprox\ algorithm for $\MaxIP$, and characterizes the best-possible additive approximations as well. We begin with formal definitions of these two standard types of approximation:

\begin{itemize}
	\item We say an algorithm $\alg$ for $\MaxIP_{n,d}$ ($\IntMaxIP_{n,d}$) is $t$-multiplicative-approximating, if for all $A,B$, $\alg$ outputs a value $\WOPT(A,B)$ such that $\WOPT(A,B) \in \left[\OPT(A,B),\OPT(A,B) \cdot t\right]$.
	
	\item We say an algorithm $\alg$ for $\MaxIP_{n,d}$ ($\IntMaxIP_{n,d}$) is $t$-additive-approximating, if for all $A,B$, $\alg$ outputs a value $\WOPT(A,B)$ such that $|\WOPT(A,B) - \OPT(A,B) | \le t$.
	
	\item To avoid ambiguity, we call an algorithm computing $\OPT(A,B)$ exactly an \emph{exact} algorithm for $\MaxIP_{n,d}$ ($\IntMaxIP_{n,d}$).
\end{itemize}

\paragraph*{Multiplicative Approximations for $\MaxIP$.} In the multiplicative case, our characterization (formally stated below) basically says that there is a $t$-\Mapprox\ $n^{2-\Omega(1)}$-time algorithm for $\MaxIP_{n, d}$ if and only if $t = \left(d/\log n\right)^{\Omega(1)}$. Note that in the following theorem we require $d = \omega(\log n)$, since in the case of $d = O(\log n)$, there \emph{are} $n^{2 - \eps}$-time algorithms for exact $\MaxIP_{n,d}$~\cite{AW15,alman2016polynomial}.

\begin{theo}~\label{theo:Max-IP-M}
	Letting $\omega(\log n) < d < n^{o(1)}$ and $t \ge 2$,\footnote{Note that $t$ and $d$ are both functions of $n$, we assume they are computable in $n^{o(1)}$ time throughout this paper for simplicity.} the following holds:
	
	\begin{enumerate}
		\item There is an $n^{2 - \Omega(1)}$-time $t$-\Mapprox\ algorithm for $\MaxIP_{n,d}$ if 
		$$
		t = \left(d/\log n\right)^{\Omega(1)},
		$$
		and under SETH (or OVC), there is no $n^{2 - \Omega(1)}$-time $t$-\Mapprox\ algorithm for $\MaxIP_{n,d}$ if
		$$
		t = \left(d/\log n\right)^{o(1)}.
		$$
		\item Moreover, let $\eps = \min\left( \frac{\log t}{\log (d/\log n)}, 1 \right)$. There are $t$-\Mapprox\ deterministic algorithms for $\MaxIP_{n,d}$ running in time
		\[
		O\left(n^{2 + o(1) - 0.31 \cdot \frac{1}{\eps^{-1} + \frac{0.31}{2}}} \right) = O\left( n^{2 + o(1) - \Omega(\eps)} \right)
		\]
		or time
		\[
		O\left(n^{2 - 0.17 \cdot \frac{1}{\eps^{-1} + \frac{0.17}{2}}} \cdot \polylog(n) \right) = O\left( n^{2 - \Omega(\eps)} \cdot \polylog(n) \right).
		\]
	\end{enumerate}
\end{theo}

\begin{rem}
	The first algorithm is slightly faster, but only truly quadratic when $\eps = \Omega(1)$, while the second algorithm still gets a non-trivial speed up over the brute force algorithm as long as $\eps = \omega(\log\log n/\log n)$.
\end{rem}

We remark here that the above algorithms indeed work for the case where the sets consisting of non-negative reals (i.e., $\posMaxIP$):

\begin{cor}\label{cor:gen-posR}
	Assuming $\omega(\log n) < d < n^{o(1)}$ and letting $\eps = \min\left( \frac{\log t}{\log (d/\log n)}, 1 \right)$, there is a $t$-\Mapprox\ deterministic algorithm for $\posR\text{-}\MaxIP_{n,d}$ running in time
	\[
	O\left( n^{2 - \Omega(\eps)} \cdot \polylog(n) \right).
	\]
\end{cor}

The lower bound is a direct corollary of the new improved $\MA$ protocols for Set-Disjointness from~\cite{Rubinstein2017closest}, which is based on Algebraic Geometry codes. Together with the framework of~\cite{ARW17-proceedings}, that $\MA$-protocol implies a reduction from $\OV$ to approximating $\MaxIP$. 

Our upper bounds are application of the polynomial method~\cite{williams2014faster,abboud2015more}: defining appropriate sparse polynomials for approximating $\MaxIP$ on small groups of vectors, and use fast matrix multiplication to speed up the evaluation of these polynomials on many pairs of points.

Via the known reduction from $\MaxIP$ to LCS-Pair in~\cite{ARW17-proceedings}, we also obtain a more refined lower bound for approximating the LCS Closest Pair problem (defined below).

\begin{defi}[LCS Closest Pair]
	The $\LCSP_{n,d}$ problem is: \emph{given two sets $A,B$ of $n$ strings from $\Sigma^{d}$ ($\Sigma$ is a finite alphabet), determine}
	\[
	\max_{a \in A,b \in B} \textsf{LCS}(a,b),
	\]
	where $\textsf{LCS}(a,b)$ is the length of the longest common subsequence of strings $a$ and $b$.
\end{defi}

\begin{cor}[Improved Inapproximability for $\LCSP$]~\label{cor:LCSP}
	Assuming SETH (or OVC), for every $t \ge 2$, $t$-\Mapprox\ $\LCSP_{n,d}$ requires $n^{2 - o(1)}$ time, if $d = t^{\omega(1)} \cdot \log^5 n$.
\end{cor}

\paragraph*{A Different Approach Based on Approximate Polynomial for $\OR$.} Making use of the $O(\sqrt{n})$-degree approximate polynomial for $\OR$~\cite{buhrman1999bounds,de2008note}, we also give a completely different proof for the hardness of multiplicative approximation to $\{-1,1\}$-$\MaxIP$.\footnote{That is, $\MaxIP$ with sets $A$ and $B$ being $n$ vectors from $\{-1,1\}^d$.} Lower bound from that approach is inferior to Theorem~\ref{theo:Max-IP-M}: in particular, \emph{it cannot achieve a characterization}.

It is asked in~\cite{ARW17-proceedings} that whether we can make use of the $O(\sqrt{n})$ $\BQP$ communication protocol for Set-Disjointness~\cite{buhrman1998quantum} to prove conditional lower bounds. Indeed, that quantum communication protocol is based on the $O(\sqrt{n})$-time quantum query algorithm for $\OR$ (Grover's algorithm~\cite{grover1996fast}), which induces the needed approximate polynomial for $\OR$. Hence, the following theorem in some sense answers their question in the affirmative:

\begin{theo}[Informal]\label{theo:informal}
	Assuming SETH (or OVC), there is no $n^{2 - \Omega(1)}$ time $n^{o(1)}$-\Mapprox\ algorithm for $\pnMaxIP_{n,n^{o(1)}}$.
\end{theo}

The full statement can be found in Theorem~\ref{theo:reduction} and Theorem~\ref{theo:quantum-based-lowb}.

\paragraph*{Additive Approximations for $\MaxIP$.}

Our characterization for additive approximations to $\MaxIP$ says that there is a $t$-\Aapprox\ $n^{2-\Omega(1)}$-time algorithm for $\MaxIP_{n, d}$ if and only if $t = \Omega(d)$.

\begin{theo}~\label{theo:Max-IP-A}
	Letting $\omega(\log n) < d < n^{o(1)}$ and $0 \le t \le d$, the following holds:
	
	\begin{enumerate}
		\item There is an $n^{2 - \Omega(1)}$-time $t$-\Aapprox\ algorithm for $\MaxIP_{n,d}$ if
		$$
		t = \Omega(d),
		$$
		and under SETH (or OVC), there is no $n^{2 - \Omega(1)}$-time $t$-\Aapprox\ algorithm for $\MaxIP_{n,d}$ if
		$$
		t = o(d).
		$$
		\item Moreover, letting $\eps = \frac{t}{d}$, there is an
		\[
		O\left(n^{2 - \Omega(\eps^{1/3}/\log\eps^{-1})}\right)
		\]
		time, $t$-\Aapprox\ randomized algorithm for $\MaxIP_{n,d}$ when $\eps \gg \log^6\log n / \log^3 n$.
	\end{enumerate}
\end{theo}

The lower bound above is already established in~\cite{Rubinstein2017closest}, while the upper bound works by reducing the problem to the $d = O(\log n)$ case via random-sampling coordinates, and solving the reduced problem via known methods~\cite{AW15,alman2016polynomial}.

\begin{rem}
	We want to remark here that the lower bounds for approximating $\MaxIP$ are direct corollaries of the new $\MA$ protocols for Set-Disjointness in~\cite{Rubinstein2017closest}. Our main contribution is providing the complementary \emph{upper bounds} to show that these lower bounds are indeed \emph{tight} assuming $\SETH$.
\end{rem}

\paragraph*{\AllPairMaxIP.} Finally, we remark here that our algorithms (with slight adaptions) also work for the following stronger problem\footnote{Since $\AllPairMaxIP$ is stronger than $\MaxIP$, lower bounds for $\MaxIP$ automatically apply for $\AllPairMaxIP$.}: $\AllPairMaxIP_{n,d}$, in which we are given two sets $A$ and $B$ of $n$ vectors from $\{0,1\}^{d}$, and for each $x \in A$ we must compute $\OPT(x,B) := \max_{y \in B} x \cdot y$. An algorithm is $t$-\Mapprox\ (\Aapprox) for $\AllPairMaxIP$ if for all $\OPT(x,B)$'s, it computes corresponding approximating answers.

\begin{cor}\label{cor:All-Pair-Max-IP}
	Suppose $ \omega(\log n) < d < n^{o(1)}$, and let 
	\[
	\eps_M := \min\left( \frac{\log t}{\log (d/\log n)}, 1 \right) \text{ and } \eps_A := \frac{\min(t,d)}{d}.
	\]
	There is an $n^{2 - \Omega(\eps_M)} \polylog(n)$ time $t$-\Mapprox\ algorithm and an $n^{2 - \Omega(\eps_A^{1/3}/\log \eps_A^{-1})}$ time $t$-\Aapprox\ algorithm for $\AllPairMaxIP_{n,d}$, when $\eps_A \gg \log^6\log n / \log^3 n$.
\end{cor}

\subsection*{Hardness of Exact $\IntMaxIP$ in $2^{O(\logstar n)}$ Dimensions}

Thirdly, we show that $\IntMaxIP$ is hard to solve in $n^{2-\Omega(1)}$ time, even with $2^{O(\logstar n)}$-dimensional vectors:

\begin{theo}\label{theo:hard-Int-Max-IP}
Assuming SETH (or OVC), there is a constant $c$ such that any exact algorithm for $\IntMaxIP_{n,d}$ for $d = c^{\log^*n}$ dimensions requires $n^{2-o(1)}$ time, with vectors of $O(\log n)$-bit entries.
\end{theo}


As direct corollaries of the above theorem, using reductions implicit in~\cite{Wil18}, we also conclude hardness for $\ell_2$-Furthest Pair and Bichromatic $\ell_2$-Closest Pair under SETH (or OVC) in $2^{O(\log^*n)}$ dimensions.

\begin{theo}[Hardness of $\ell_2$-Furthest Pair in $c^{\log^* n}$ Dimensions] \label{theo:l-2-furthest-pair}
	Assuming SETH (or OVC), there is a constant $c$ such that $\ell_2$-Furthest Pair in $c^{\log^*n}$ dimensions requires $n^{2-o(1)}$ time, with vectors of $O(\log n)$-bit entries.
\end{theo}

\begin{theo}[Hardness of Bichromatic $\ell_2$-Closest Pair in $c^{\log^* n}$ Dimensions] \label{theo:bi-closest-pair}
	Assuming SETH (or OVC), there is a constant $c$ such that Bichromatic $\ell_2$-Closest Pair in $c^{\log^*n}$ dimensions requires $n^{2-o(1)}$ time, with vectors of $O(\log n)$-bit entries.
\end{theo}

The above lower bounds on $\ell_2$-Furthest Pair and Bichromatic $\ell_2$-Closest Pair are in sharp contrast with the case of \emph{$\ell_2$-Closest Pair}, which can be solved in $2^{O(d)} \cdot n \log^{O(1)} n$ time~\cite{bentley1976divide,khuller1995simple,dietzfelbinger1997reliable}.

\subsection*{Improved Dimensionality Reduction for $\OV$ and Hopcroft's Problem}

Our hardness of $\IntMaxIP$ is established by a reduction from Hopcroft's problem, whose hardness is in turn derived from the following significantly improved dimensionality reduction for $\OV$.
\begin{lemma}[Improved Dimensionality Reduction for $\OV$]\label{lm:dim-reduction-OV}
	Let $1 \le \ell \le d$. There is an 
	$$
	O\left(n \cdot \ell^{O(6^{\log^*d} \cdot (d/\ell))} \cdot \operatorname*{poly}(d) \right)\text{-time}
	$$
	reduction from $\OV_{n,d}$ to $\ell^{O(6^{\log^*d} \cdot (d/\ell))}$ instances of $\Hopcroft_{n,\ell + 1}$, with vectors of entries with bit-length $O\left(d/\ell \cdot \log \ell \cdot 6^{\log^* d}\right)$.
\end{lemma}

\paragraph*{Comparison with~\cite{Wil18}.} Comparing to the old construction in~\cite{Wil18}, our reduction here is more efficient when $\ell$ is much smaller than $d$ (which is the case we care about). That is, in~\cite{Wil18}, $\OV_{n,d}$ can be reduced to $d^{d/\ell}$ instances of $\IntOV_{n,\ell+1}$, while we get $\left\{ \ell^{6^{\log^{*}d}} \right\}^{d/\ell}$ instances in our improved one. So, for example, when $\ell = 7^{\log^*d}$, the old reduction yields $d^{d / 7^{\log^*d}} = n^{\omega(1)}$ instances (recall that $d = c \log n$ for an arbitrary constant $c$), while our improved one yields only $n^{o(1)}$ instances, each with $2^{O(\logstar n)}$ dimensions.

From Lemma~\ref{lm:dim-reduction-OV}, the following theorem follows in the same way as in~\cite{Wil18}.

\begin{theo}[Hardness of Hopcroft's Problem in $c^{\log^* n}$ Dimensions] \label{theo:Hopcroft}
	Assuming SETH (or OVC), there is a constant $c$ such that $\Hopcroft_{n,c^{\logstar n}}$ with vectors of $O(\log n)$-bit entries requires $n^{2-o(1)}$ time.
\end{theo}

\subsection*{Connection between $\IntMaxIP$ lower bounds and $\NP \cdot \UPP$ communication protocols}

We also show a new connection between $\IntMaxIP$ and a special type of communication protocol. Let us first recall the Set-Disjointness problem:

\begin{defi}[Set-Disjointness]
	Let $n \in \mathbb{N}$, in Set-Disjointness ($\DISJ_{n}$), Alice holds a vector $X \in \{0,1\}^n$, Bob holds a vector $Y \in \{0,1\}^n$, and they want to determine whether $X \cdot Y = 0$.
\end{defi}

Recall that in~\cite{ARW17-proceedings}, the hardness of approximating $\MaxIP$ is established via a connection to $\MA$ communication protocols (in particular, a fast $\MA$ communication protocol for Set-Disjointness). Our lower bound for (exact) $\IntMaxIP$ can also be connected to similar $\NP \cdot \UPP$ protocols (note that $\MA = \NP \cdot \promiseBPP$). 

Formally, we define $\NP \cdot \UPP$ protocols as follows:

\begin{defi}\label{defi:NPUPP}
	For a problem $\Pi$ with inputs $x,y$ of length $n$ (Alice holds $x$ and Bob holds $y$), we say a communication protocol is an \emph{$(m,\ell)$-efficient $\NP \cdot \UPP$ communication protocol} if the following holds:
	
	\begin{itemize}
		\item There are three parties Alice, Bob and Merlin in the protocol.
		
		\item Merlin sends Alice and Bob an advice string $z$ of length $m$, which is a function of $x$ and $y$.
		
		\item Given $y$ and $z$, Bob sends Alice $\ell$ bits, and Alice decides to accept or not.\footnote{In $\UPP$, actually one-way communication is equivalent to the seemingly more powerful one in which they communicate~\cite{paturi1986probabilistic}.} They have an unlimited supply of private random coins (not public, which is important) during their conversation. The following conditions hold:
		
		\begin{itemize}
			\item If $\Pi(x,y) = 1$, then there is an advice $z$ from Merlin such that Alice accepts with probability $\ge 1/2$.
			\item Otherwise, for all possible advice strings from Merlin, Alice accepts with probability $< 1/2$.
		\end{itemize}
		
	\end{itemize}
	
	Moreover, we say the protocol is $(m,\ell)$-computational-efficient, if in addition the probability distributions of both Alice and Bob's behavior can be computed in $\poly(n)$ time given their input and the advice.
\end{defi}

Our new reduction from $\OV$ to $\MaxIP$ actually implies a super-efficient $\NP \cdot \UPP$ protocol for Set-Disjointness.

\begin{theo}\label{theo:NPUPP-for-DISJ}
	For all $1 \le \alpha \le n$, there is an
	\[
	\left( \alpha \cdot 6^{\logstar n} \cdot (n/2^\alpha), O(\alpha) \right)\text{-computational-efficient}
	\]
	$\NP \cdot \UPP$ communication protocol for $\DISJ_{n}$.
\end{theo}

For example, when $\alpha = 3\logstar n$, Theorem~\ref{theo:NPUPP-for-DISJ} implies there is an $O(o(n),O(\logstar n))$-computational-efficient $\NP \cdot \UPP$ communication protocol for $\DISJ_{n}$. Moreover, we show that if the protocol of Theorem~\ref{theo:NPUPP-for-DISJ} can be improved a little (removing the $6^{\logstar n}$ term), we would obtain the desired  hardness for $\IntMaxIP$ in $\omega(1)$-dimensions.

\begin{theo}\label{theo:better-imply-MaxIP}
	Assuming SETH (or OVC), if there is an increasing and unbounded function $f$ such that for all $1 \le \alpha \le n$, there is an
	\[
	\left( n / f(\alpha), \alpha \right)\text{-computational-efficient}
	\]
	$\NP \cdot \UPP$ communication protocol for $\DISJ_{n}$, then $\IntMaxIP_{n,\omega(1)}$ requires $n^{2 - o(1)}$ time with vectors of $\polylog(n)$-bit entries. The same holds for $\ell_2$-Furthest Pair and Bichromatic $\ell_2$-Closest Pair.
\end{theo}

\subsection*{Improved $\MA$ Protocols for Set-Disjointness}
\newcommand{\InProd}{\textsf{IP}}

Finally, we also obtain a new $\MA$ protocol for Set-Disjointness, which improves on the previous $O(\sqrt{n} \log n)$ protocol in~\cite{AW09-algebrization}, and is closer to the $\Omega(\sqrt{n})$ lower bound by~\cite{klauck2003rectangle}. Like the protocol in~\cite{AW09-algebrization}, our new protocol also works for the following slightly harder problem Inner Product.

\begin{defi}[Inner Product]
	Let $n \in \mathbb{N}$, in Inner Product ($\InProd_{n}$), Alice holds a vector $X \in \{0,1\}^n$, Bob holds a vector $Y \in \{0,1\}^n$, and they want to compute $X \cdot Y$.
\end{defi}

\begin{theo}\label{theo:improved-MA}
	There is an $\MA$ protocol for $\DISJ_{n}$ and $\InProd_{n}$ with communication complexity
	$$
	O\left(\sqrt{n\log n\log\log n}\right).
	$$
\end{theo}

In~\cite{Rubinstein2017closest}, the author asked whether the $\MA$ communication complexity of $\DISJ$ ($\InProd$) is $\Theta(\sqrt{n})$ or $\Theta(\sqrt{n\log n})$, and suggested that $\Omega(n \log n)$ may be necessary for $\InProd$. Our result makes progress on that question by showing that the true complexity lies between $\Theta(\sqrt{n})$ and $\Theta(\sqrt{n\log n\log\log n})$.

\subsection{Intuition for Dimensionality Self Reduction for $\OV$}

The $2^{O(\logstar n)}$ factor in Lemma~\ref{lm:dim-reduction-OV} is not common in theoretical computer science\footnote{Other examples include an $O\big(2^{O(\logstar n)} n^{4/3} \big)$ algorithm for $\IntOV_{n,3}$~\cite{matouvsek1993range}, $O\big(2^{O(\logstar n)} n \log n\big)$ algorithms (F\"urer's algorithm with its modifications) for Fast Integer Multiplication~\cite{furer2009faster,covanov2015fast,harvey2016even} and an old $O(n^{d/2} 2^{O(\logstar n)})$ time algorithm for Klee's measure problem~\cite{chan2008slightly}.}, and our new reduction for $\OV$ is considerably more complicated than the polynomial-based construction from~\cite{Wil18}. Hence, it is worth discussing the intuition behind Lemma~\ref{lm:dim-reduction-OV}, and the reason why we get a factor of $2^{O(\logstar n)}$.

\paragraph*{A Direct Chinese Remainder Theorem Based Approach.} We first discuss a direct reduction based on the \emph{Chinese Remainder Theorem} (CRT) (see~Theorem~\ref{theo:CRR} for a formal definition). CRT says that given a collection of primes $q_1,\dotsc,q_b$, and a collection of integers $r_1,\dotsc,r_b$, there exists a unique integer $t = \CRR(\{r_i\};\{q_i\})$ such that $ t \equiv r_i \pmod{q_i}$ for each $i \in [b]$ (CRR stands for \emph{Chinese Remainder Representation}).

\newcommand{\varphiblock}{\varphi_{\textsf{block}}}
\newcommand{\psiblock}{\psi_{\textsf{block}}}

Now, let $b,\ell \in \mathbb{N}$, suppose we would like to have a dimensionality reduction $\varphi$ from $\{0,1\}^{b \cdot \ell}$ to $\mathbb{Z}^{\ell}$. We can partition an input $x \in \{0,1\}^{b \cdot \ell}$ into $\ell$ blocks, each of length $b$, and represent each block via CRT: that is, for a block $z \in \{0,1\}^{b}$, we map it into a single integer $ \varphiblock(z) := \CRR(\{z_i\};\{q_i\})$, and the concatenations of $\varphiblock$ over all blocks of $x$ is $\varphi(x) \in \mathbb{Z}^\ell$.

The key idea here is that, for $z,z' \in \{0,1\}^b$, $\varphiblock(z) \cdot \varphiblock(z') \pmod{q_i}$ is simply $z_i \cdot z'_i$. That is, the multiplication between two \emph{integers} $\varphiblock(z) \cdot \varphiblock(z')$ simulates the coordinate-wise multiplication between two \emph{vectors} $z$ and $z'$!

Therefore, if we make all primes $q_i$ larger than $\ell$, we can in fact determine $x \cdot y$ from $\varphi(x) \cdot \varphi(y)$, by looking at $\varphi(x) \cdot \varphi(y) \pmod{q_i}$ for each $i$. That is,
\[
x \cdot y = 0 \Leftrightarrow \varphi(x) \cdot \varphi(y) \equiv 0 \pmod{q_i} \quad\text{for all $i$.}
\]

Hence, let $V$ be the set of all integer $0 \le v \le \ell \cdot \left(\prod_{i=1}^{b} q_i \right)^2$ that $v \equiv 0 \pmod{q_i}$ for all $i \in [b]$, we have
\[
x \cdot y = 0 \Leftrightarrow \varphi(x) \cdot \varphi(y) \in V.
\]

The reduction is completed by enumerating all integers $v \in V$, and appending corresponding values to make $\varphi_A(x) = [\varphi(x),-1]$ and $\varphi_B(y) = [\varphi(y),v]$ (this step is from~\cite{Wil18}). 

Note that a nice property for $\varphi$ is that each $\varphi(x)_i$ only depends on the $i$-th block of $x$, and the mapping is the same on each block ($\varphiblock$); we call this the \emph{block mapping property}.

\paragraph*{Analysis of the Direct Reduction.} To continue building intuition, let us analyze the above reduction. The size of $V$ is the number of $\Hopcroft_{n,\ell + 1}$ instances we create, and $|V| \ge \prod_{i=1}^{b} q_i$. These primes $q_i$ have to be all distinct, and it follows that $\prod_{i=1}^{b} q_i$ is $b^{\Theta(b)}$. Since we want to create at most $n^{o(1)}$ instances (or $n^\eps$ for arbitrarily small $\eps$), we need to set $b \le \log n/\log\log n$. Moreover, to base our hardness on OVC which deals with $c \log n$-dimensional vectors, we need to set $b \cdot \ell = d = c \cdot \log n$ for an arbitrary constant $c$. Therefore, we must have $\ell \ge \log\log n$, and the above reduction only obtains the same hardness result as~\cite{Wil18}.

\paragraph*{Key Observation: ``Most Space Modulo $q_i$'' is Actually Wasted.}
To improve the above reduction, we need to make $|V|$ smaller. Our key observation about $\varphi$ is that, for the primes $q_i$'s, they are mostly larger than $b \gg \ell$, but $\varphi(x) \cdot \varphi(y) \in \{0,1,\dotsc,\ell \} \pmod{q_i}$ for all these $q_i$'s. Hence, \emph{``most space modulo $q_i$'' is actually wasted.}


\paragraph*{Make More ``Efficient'' Use of the ``Space'': Recursive Reduction.} 
Based on the previous observation, we want to use the ``space modulo $q_i$'' more efficiently. It is natural to consider a \emph{recursive reduction}. We will require all our primes $q_i$'s to be larger than $b$. Let $\bm$ be a very small integer compared to $b$, and let $\psi: \{0,1\}^{\bm \cdot \ell} \to \mathbb{Z}^{\ell}$ with a set $V_\psi$ and a block mapping $\psiblock$ be a similar reduction on a much smaller input: for $x,y \in \{0,1\}^{\bm \cdot \ell}$, $ x \cdot y = 0 \Leftrightarrow \psi(x) \cdot \psi(y) \in V_\psi$. We also require here that $\psi(x) \cdot \psi(y) \le b$ for all $x$ and $y$.

For an input $x \in \{0,1\}^{b \cdot \ell}$ and a block $z \in \{0,1\}^b$ of $x$, our key idea is to partition $z$ again into $b/\bm$ ``micro'' blocks each of size $\bm$. And for a block $z$ in $x$, let $z^{1},\dotsc,z^{b/\bm}$ be its $b/\bm$ micro blocks, we map $z$ into an integer $ \varphiblock(z) := \CRR( \{ \psiblock(z_i) \}_{i=1}^{b/\bm} ; \{q_i\}_{i=1}^{b/\bm} )$.

Now, given two blocks $z,z' \in \{0,1\}^b$, we can see that
\[
\varphiblock(z) \cdot \varphiblock(z') \equiv \psiblock(z_i) \cdot \psiblock(z'_i) \pmod{q_i}.
\]

That is, $\varphi(x) \cdot \varphi(y) \pmod{q_i}$ in fact is equal to $\psi(x^{[i]}) \cdot \psi(y^{[i]})$, where $x^{[i]}$ is the concatenation of the $i$-th micro blocks of $x$ in each block, and $y^{[i]}$ is defined similarly. Hence, we can determine whether $x^{[i]} \cdot y^{[i]} = 0$ from $\varphi(x) \cdot \varphi(y) \pmod{q_i}$ for all $i$, and therefore also determine whether $x \cdot y = 0$ from $\varphi(x) \cdot \varphi(y)$.

We can now observe that $|V| \le b^{\Theta(b/\bm)}$, smaller than before; thus we get an improvement, depending on how large can $\bm$ be. Clearly, the reduction $\psi$ can also be constructed from even smaller reductions, and after recursing $\Theta(\logstar n)$ times, we can switch to the direct construction discussed before. By a straightforward (but tedious) calculation, we can derive Lemma~\ref{lm:dim-reduction-OV}.

\paragraph*{High-Level Explanation on the $2^{O(\logstar n)}$ Factor.} Ideally, we want to have a reduction from $\OV$ to $\Hopcroft$ with only $\ell^{O(b)}$ instances, in other words, we want $|V| = \ell^{O(b)}$. The reason we need to pay an extra $2^{O(\logstar n)}$ factor in the exponent is as follows:

In our reduction, $|V|$ is at least $\prod_{i=1}^{b/\bm} q_i$, which is also the bound on each coordinate of the reduction: $\psi(x)_i$ equals to a $\CRR$ encoding of a vector with $\{q_i\}_{i=1}^{b/\bm}$, whose value can be as large as $\prod_{i=1}^{b/\bm} q_i - 1$. That is, all we want is to control the upper bound on the coordinates of the reduction.

Suppose we are constructing an ``outer'' reduction $\varphi : \{0,1\}^{b \cdot \ell} \to \mathbb{Z}^\ell$ from the ``micro'' reduction $\psi : \{0,1\}^{\bm \cdot \ell} \to \mathbb{Z}^\ell$ with coordinate upper bound $L_\psi$ ($\psi(x)_i \le L_\psi $), and let $L_\psi = \ell^{\kappa \cdot \bm}$ (that is, $\kappa$ is the extra factor comparing to the ideal case). Recall that we have to ensure $q_i > \psi(x) \cdot \psi(y)$ to make our construction work, and therefore we have to set $q_i$ larger than $L_\psi^2$. 

Then the coordinate upper bound for $\varphi$ becomes $L_\varphi = \prod_{i=1}^{b/\bm} q_i \ge (L_\psi)^{2 \cdot b/\bm} = \ell^{2 \kappa \cdot b}$. Therefore, we can see that after one recursion, the ``extra factor'' $\kappa$ at least doubles. Since our recursion proceeds in $\Theta(\logstar n)$ rounds, we have to pay an extra $2^{O(\logstar n)}$ factor on the exponent.

\subsection{Related Works}

\paragraph*{SETH-based Conditional Lower Bound.}
SETH is one of the most fruitful conjectures in the Fine-Grained Complexity. There are numerous conditional lower bounds based on it for problems in $\PTIME$ among different areas, including: dynamic data structures~\cite{AV14}, computational geometry~\cite{Bring14,Wil18,david2016complexity}, pattern matching~\cite{AVW14,BI15,BI16,bringmann2016dichotomy,bringman2018multivariate}, graph algorithms~\cite{RV13,GIKW17,abboud2015matching,krauthgamer2017conditional}. See~\cite{williamssome} for a very recent survey on SETH-based lower bounds (and more).

\paragraph*{Communication Complexity and Conditional Hardness.}
The connection between communication protocols (in various model) for Set-Disjointness and SETH dates back at least to~\cite{PW10}, in which it is shown that a sub-linear, computational efficient protocol for $3$-party Number-On-Forehead Set-Disjointness problem would refute SETH. And it is worth mentioning that \cite{abboud2018fast}'s result builds on the $\widetilde{O}(\log n)$ $\textsf{IP}$ communication protocol for Set-Disjointness in~\cite{AW09-algebrization}.

\paragraph*{Distributed PCP.} Using Algebraic Geometry codes,~\cite{Rubinstein2017closest} obtains a better $\MA$ protocol, which in turn improves the efficiency of the previous distributed PCP construction of~\cite{ARW17-proceedings}. He then shows the $n^{2-o(1)}$ time hardness for $1+o(1)$-approximation to Bichromatic Closest Pair and $o(d)$-additive approximation to $\MaxIP_{n,d}$ with this new technique. 

\cite{karthik2017parameterized} use the Distributed PCP framework to derive inapproximability results for $k$-Dominating Set under various assumptions. In particular, building on the techniques of~\cite{Rubinstein2017closest}, it is shown that under SETH, $k$-Dominating Set has no $(\log n)^{1/\poly(k,e(\eps))}$ approximation in $n^{k-\eps}$ time\footnote{where $e: \posR \to \mathbb{N}$ is some function}. 

\paragraph*{Hardness of Approximation in $\PTIME$.} Making use of Chebychev embeddings, \cite{ahle2016complexity} prove a $2^{\Omega\left(\frac{\sqrt{\log n}}{\log\log n}\right)}$ inapproximability lower bound on $\pnMaxIP$.\footnote{which is improved by Theorem~\ref{theo:informal}} \cite{abboud2017towards} take an approach different from Distributed PCP, and shows that under certain complexity assumptions, $\textsf{LCS}$ does not have a \emph{deterministic} $1+o(1)$-approximation in $n^{2 -\eps}$ time. They also establish a connection with circuit lower bounds and show that the existence of such a \emph{deterministic} algorithm implies $\mathsf{E}^{\NP}$ does not have non-uniform linear-size Valiant Series Parallel circuits. In~\cite{abboud2018fast}, it is improved to that any constant factor approximation deterministic algorithm for $\textsf{LCS}$ in $n^{2 - \eps}$ time implies that $\mathsf{E}^\NP$ does not have non-uniform linear-size $\textsf{NC}^1$ circuits. See~\cite{ARW17-proceedings} for more related results in hardness of approximation in $\PTIME$.

\subsection*{Organization of the Paper}
In Section~\ref{sec:prelim}, we introduce the needed preliminaries for this paper. In Section~\ref{sec:approx}, we prove our characterizations for approximating $\MaxIP$ and other related results. In Section~\ref{sec:exact}, we prove $2^{O(\logstar n)}$ dimensional hardness for $\IntMaxIP$ and other related problems. In Section~\ref{sec:NP-UPP}, we establish the connection between $\NP \cdot \UPP$ communication protocols and SETH-based lower bounds for exact $\IntMaxIP$. In Section~\ref{sec:MA}, we present the $O\left(\sqrt{n\log n \log\log n} \right)$ $\MA$ protocol for Set-Disjointness.

	\section{Preliminaries}\label{sec:prelim}

We begin by introducing some notation. For an integer $d$, we use $[d]$ to denote the set of integers from $1$ to $d$. For a vector $u$, we use $u_{i}$ to denote the $i$-th element of $u$.

We use $\log(x)$ to denote the logarithm of $x$ with respect to base $2$ with ceiling as appropriate, and $\ln(x)$ to denote the natural logarithm of $x$.

In our arguments, we use the iterated logarithm function $\log^*(n)$, which is defined recursively as follows:
$$
\log^*(n) := \begin{cases}
0 &\quad n \le 1; \\
\log^*(\log n) + 1  &\quad n > 1.
\end{cases}
$$

\subsection{Fast Rectangular Matrix Multiplication}

Similar to previous algorithms using the polynomial method, our algorithms make use of the algorithms for fast rectangular matrix multiplication. 

\begin{theo}[\cite{gall2018improved}]\label{theo:fast-matrix-mult}
	There is an $N^{2 + o(1)}$ time algorithm for multiplying two matrices $A$ and $B$ with size $N \times N^{\alpha}$ and $N^{\alpha} \times N$, where $\alpha > 0.31389$.
\end{theo}

\begin{theo}[\cite{coppersmith1982rapid}]\label{theo:fast-matrix-mult-polylog}
	There is an $N^{2} \cdot \polylog(N)$ time algorithm for multiplying two matrices $A$ and $B$ with size $N \times N^{\alpha}$ and $N^{\alpha} \times N$, where $\alpha > 0.172$.
\end{theo}

\subsection{Number Theory}

Here we recall some facts from number theory. In our reduction from $\OV$ to $\Hopcroft$, we will apply the famous prime number theorem, which supplies a good estimate of the number of primes smaller than a certain number. See e.g.~\cite{apostol2013introduction} for a reference on this.


\begin{theo}[Prime Number Theorem]
	Let $\pi(n)$ be the number of primes $\le n$, then we have
	$$
	\lim_{n \to \infty} \frac{\pi(n)}{n / \ln n} = 1.
	$$
\end{theo}

From a simple calculation, we obtain:

\begin{lemma}\label{lm:many-primes}
	There are $10n$ distinct primes in $[n+1,n^2]$ for a large enough $n$.
\end{lemma}

\begin{proof}
	For a large enough $n$, from the prime number theorem, the number of primes in $[n+1,n^2]$ is equal to
	$$
	\pi(n^2) - \pi(n) \sim n^2 / 2 \ln n - n/\ln n \gg 10n.
	$$        
\end{proof}

Next we recall the Chinese remainder theorem, and Chinese remainder representation.

\begin{theo}\label{theo:CRR}
	Given $d$ pairwise co-prime integers $q_1,q_2,\dotsc,q_{d}$, and $d$ integers $r_1,r_2,\dotsc,r_d$, there is exactly one integer $0 \le t  < \prod_{i=1}^{d} q_i$ such that
	$$
	t \equiv r_i \pmod{q_i} \quad\text{for all $i \in [d]$.}
	$$
	We call this $t$ the Chinese remainder representation (or the CRR encoding) of the $r_i$'s (with respect to these $q_i$'s). We also denote 
	$$
	t = \CRR(\{r_i\} ; \{q_i\})
	$$
	for convenience. We sometimes omit the sequence $\{q_i\}$ for simplicity, when it is clear from the context.
	
	Moreover, $t$ can be computed in polynomial time with respect to the total bits of all the given integers.
\end{theo}

\subsection{Communication Complexity}


In our paper we will make use of a certain kind of $\MA$ protocol, we call them $(m,r,\ell,s)$-efficient protocols\footnote{Our notations here are adopted from~\cite{karthik2017parameterized}. They also defined similar $k$-party communication protocols, while we only discuss $2$-party protocols in this paper.}.

\begin{defi}
	We say an $\MA$ Protocol is $(m,r,\ell,s)$-efficient for 
	a communication problem, if in the protocol:
	
	\begin{itemize}
		\item There are three parties Alice, Bob and Merlin in the protocol, Alice holds input $x$ and Bob holds input $y$.
		
		\item Merlin sends an advice string $z$ of length $m$ to Alice, which is a function of $x$ and $y$.
		
		\item Alice and Bob jointly toss $r$ coins to obtain a random string $w$ of length $r$.
		
		\item Given $y$ and $w$, Bob sends Alice a message of length $\ell$.
		
		\item After that, Alice decides whether to accept or not.
		
		\begin{itemize}
			\item When the answer is yes, Merlin has exactly one advice such that Alice always accept.
			
			\item When the answer is no, or Merlin sends the wrong advice, Alice accepts with probability at most $s$.
		\end{itemize}
	\end{itemize}
\end{defi}

\subsection{Derandomization}

We make use of expander graphs to reduce the amount of random coins needed in one of our communication protocols. We abstract the following result for our use here.

\begin{theo}[see e.g. Theorem 21.12 and Theorem~21.19 in~\cite{AB09-book}]\label{theo:derand}
	Let $m$ be an integer, and set $B \subseteq [m]$. Suppose $|B| \ge m/2$. There is a universal constant $c_1$ such that for all $\eps < 1/2$, there is a $\poly(\log m,\log \eps^{-1})$-time computable function $\mathcal{F} :　\{0,1\}^{\log m + c_1 \cdot \log \eps^{-1}} \to [m]^{c_1 \cdot \log \eps^{-1}}$, such that
	\[
	\Pr_{w \in \{0,1\}^{\log m + c_1 \cdot \log \eps^{-1}}}\left[ a \notin B \text{ for all $a \in \mathcal{F}(w)$}  \right] \le \eps,
	\]
	here $a \in \mathcal{F}(w)$ means $a$ is one of the element in the sequence $\mathcal{F}(w)$.
\end{theo}

	\section{Hardness of Approximate $\MaxIP$}\label{sec:approx}
	
	In this section we prove our characterizations of approximating $\MaxIP$.
	
	\subsection{The Multiplicative Case}
	
	We begin with the proof of Theorem~\ref{theo:Max-IP-M}. We recap it here for convenience.
	
	\begin{reminder}{Theorem~\ref{theo:Max-IP-M}}
		Letting $\omega(\log n) < d < n^{o(1)}$ and $t \ge 2$, the following holds:
		
		\begin{enumerate}
			\item There is an $n^{2 - \Omega(1)}$-time $t$-\Mapprox\ algorithm for $\MaxIP_{n,d}$ if 
			$$
			t = \left(d/\log n\right)^{\Omega(1)},
			$$
			and under SETH (or OVC), there is no $n^{2 - \Omega(1)}$-time $t$-\Mapprox\ algorithm for $\MaxIP_{n,d}$ if
			$$
			t = \left(d/\log n\right)^{o(1)}.
			$$
			\item Moreover, let $\eps = \min\left( \frac{\log t}{\log (d/\log n)}, 1 \right)$. There are $t$-\Mapprox\ deterministic algorithms for $\MaxIP_{n,d}$ running in time
			\[
			O\left(n^{2 + o(1) - 0.31 \cdot \frac{1}{\eps^{-1} + \frac{0.31}{2}}} \right) = O\left( n^{2 + o(1) - \Omega(\eps)} \right)
			\]
			or time
			\[
			O\left(n^{2 - 0.17 \cdot \frac{1}{\eps^{-1} + \frac{0.17}{2}}} \cdot \polylog(n) \right) = O\left( n^{2 - \Omega(\eps)} \cdot \polylog(n) \right).
			\]
		\end{enumerate}
	\end{reminder}

	In Lemma~\ref{lm:algo-Max-IP-M}, we construct the desired approximate algorithm and in Lemma~\ref{lm:lowb-Max-IP-M} we prove the lower bound.
	
	
	\subsubsection*{The Algorithm}	
	First we need the following simple lemma, which says that the $k$-th root of the sum of the $k$-th powers of non-negative reals gives a good approximation to their maximum.

\begin{lemma}\label{lm:simple-approx}
	Let $S$ be a set of non-negative real numbers, $k$ be an integer, and $x_{max} := \max_{x \in S} x$. We have
	$$
	\left(\sum_{x\in S} x^k \right)^{1/k} \in \left[x_{max},x_{max} \cdot |S|^{1/k}\right].
	$$
\end{lemma}

\begin{proof}
	Since
	$$
	\left( \sum_{x \in S} x^k \right) \in \left[x_{max}^k, |S| \cdot x_{max}^k \right],
	$$
	the lemma follows directly by taking the $k$-th root of both sides.
	
\end{proof}

\begin{lemma}~\label{lm:algo-Max-IP-M}
	Assuming $\omega(\log n) < d < n^{o(1)}$ and letting $\eps = \min\left( \frac{\log t}{\log (d/\log n)}, 1 \right)$, there are $t$-\Mapprox\ deterministic algorithms for $\MaxIP_{n,d}$ running in time
		\[
		O\left(n^{2 + o(1) - 0.31 \cdot \frac{1}{\eps^{-1} + \frac{0.31}{2}}} \right) = O\left( n^{2 + o(1) - \Omega(\eps)} \right)
		\]
		or time
		\[
		O\left(n^{2 - 0.17 \cdot \frac{1}{\eps^{-1} + \frac{0.17}{2}}} \cdot \polylog(n) \right) = O\left( n^{2 - \Omega(\eps)} \cdot \polylog(n) \right).
		\]
\end{lemma}

\begin{proof}
	Let $d =  c \cdot \log n$. From the assumption, we have $c = \omega(1)$, and $\eps = \min\left( \frac{\log t}{\log c}, 1 \right)$. When $\log t > \log c$, we simply use a $c$-\Mapprox\ algorithm instead, hence in the following we assume $\log t \le \log c$. We begin with the first algorithm here.
	
	\paragraph*{Construction and Analysis of the Power of Sum Polynomial $P_r(z)$.}
	Let $r$ be a parameter to be specified later and $z$ be a vector from $\{0,1\}^d$, consider the following polynomial
	$$
	P_{r}(z) := \left(\sum_{i=1}^{d} z_i \right)^{r}.
	$$
	
	Observe that since each $z_i$ takes value in $\{0,1\}$, we have $z_i^k = z_i$ for $k \ge 2$. Therefore, by expanding out the polynomial and replacing all $z_i^{k}$ with $k \ge 2$ by $z_i$, we can write $P_{r}(z)$ as
	$$
	P_{r}(z) = \sum_{S \subseteq [d], |S| \le r} c_S \cdot z_S. 
	$$
	
	In which $z_S := \prod_{i \in S} z_i$, and the $c_S$'s are the corresponding coefficients. Note that $P_{r}(z)$ has $$
	m := \sum_{k=0}^{r} \binom{d}{k} \le \left( \frac{e d}{r} \right)^{r}
	$$
	terms.
	
	Then consider $P_{r}(x,y) := P_{r}(x_1 \cdot y_1,x_2 \cdot y_2,\dotsc,x_{d}\cdot y_{d})$, plugging in $z_i := x_i \cdot y_i$, it can be written as
	$$
	P_{r}(x,y) := \sum_{S \subseteq [d], |S| \le r} c_S \cdot x_S \cdot y_S,
	$$
	where $x_S := \prod_{i \in S} x_i$, and $y_S$ is defined similarly.
	
	\paragraph*{Construction and Analysis of the Batch Evaluation Polynomial $P_r(X,Y)$.}
	Now, let $X$ and $Y$ be two sets of $b = t^{r/2}$ vectors from $\{0,1\}^{d}$, we define
	$$ 
	P_r(X,Y) := \sum_{x \in X, y \in Y} P_{r}(x,y) = \sum_{x \in X,y \in Y} (x \cdot y)^{r}.
	$$
	
	By Lemma~\ref{lm:simple-approx}, we have
	\[
	P_r(X,Y)^{1/r} \in \left[\OPT(X,Y),\OPT(X,Y) \cdot t\right],
	\]
	recall that $\OPT(X,Y) := \max_{x \in X,y \in Y} x \cdot y$.
	
	\paragraph*{Embedding into Rectangle Matrix Multiplication.} Now, for $x,y \in \{0,1\}^{d}$, we define the mapping $\phi_x(x)$ as
	\[
	\phi_x(x) := \left(c_{S_1} \cdot x_{S_1},c_{S_2} \cdot x_{S_2},\dotsc,c_{S_m} \cdot x_{S_m}\right)
	\]
	and
	\[
	\phi_y(y) := \left(y_{S_1},y_{S_2},\dotsc,y_{S_m}\right),
	\]
	where $S_1,S_2,\dotsc,S_m$ is an enumeration of all sets $S \subseteq [d]$ and $|S| \le r$.
	
	From the definition, it follows that
	$$
	\phi_x(x) \cdot \phi_y(y) = P_{r}(x,y)
	$$
	for every $x,y \in \{0,1\}^{d}$.
	
	
	Then for each $X$ and $Y$, we map them into $m$-dimensional vectors $\phi_X(X)$ and $\phi_Y(Y)$ simply by a summation:
	\[
	\phi_X(X) := \sum_{x \in X} \phi_x(x) \quad\text{and}\quad \Phi_Y(Y) := \sum_{y \in Y} \phi_y(y).
	\]
	We can see
	\[
	\phi_X(X) \cdot \phi_Y(Y) = \sum_{x \in X} \phi_x(x) \cdot \sum_{y \in Y} \phi_y(y) = \sum_{x \in X} \sum_{y \in Y} P_r(x,y) = P_r(X,Y).
	\]
	
	Given two sets $A,B$ of $n$ vectors from $\{0,1\}^d$, we split $A$ into $n/b$ sets $A_1,A_2,\dotsc,A_{n/b}$ of size $b$, and split $B$ in the same way as well. Then we construct a matrix $M_A(M_B)$ of size $n/b \times m$, such that the $i$-th row of $M_A(M_B)$ is the vector $\Phi_X(A_i)(\Phi_Y(B_i))$. After that, the evaluation of $P_r(A_i,B_j)$ for all $i,j \in [n/b]$ can be reduced to compute the matrix product $M_A \cdot M_B^{T}$. After knowing all $P_r(A_i,B_j)$'s, we simply compute the maximum of them, whose $r$-th root gives us a $t$-\Mapprox\ answer of the original problem.
	
	\paragraph*{Analysis of the Running Time.}
	Finally, we are going to specify the parameter $r$ and analyze the time complexity. In order to utilize the fast matrix multiplication algorithm from Theorem~\ref{theo:fast-matrix-mult}, we need to have
	$$
	m \le (n/b)^{0.313},
	$$
	then our running time is simply $(n/b)^{2 + o(1)} = n^{2 + o(1)} / b^2$.
	
	We are going to set $r = k \cdot \log n /\log c$, and our choice of $k$ will satisfy $k = \Theta(1)$. We have
	\[
	m \le \left(\frac{e \cdot d}{r}\right)^{r} \le \left( \frac{c \log n \cdot e}{k \cdot \log n/\log c} \right)^{k \cdot \log n/\log c},
	\]
	and therefore
	\[
	\log m \le k \cdot \log n \left[ \log \frac{c\log c}{k} + 1 \right] \Big/ \log c.
	\]
	
	Since $c = \omega(1)$ and $k = \Theta(1)$, we have
	\[
	\log m \le (1 + o(1)) \cdot k \log n = k \log n + o(\log n).
	\]
	
	Plugging in, we have
	\begin{align*}
	&m \le (n/b)^{0.313} \\
	\impliedby& \log m \le 0.313 \cdot (\log n - \log b)\\
	\impliedby& k \log n \le 0.31 \cdot (\log n - \log b)\\
	\impliedby& 0.31 \cdot (r/2) \cdot \log t + k \log n \le 0.31 \log n \tag{$b = t^{r/2}$}\\
	\impliedby& \frac{\log n}{\log c} \cdot k \cdot \log t \cdot \frac{0.31}{2} + k \log n \le 0.31 \log n \tag{$r = k \cdot \log n /\log c$}\\
	\impliedby& k \cdot \left\{ 1 + \frac{\log t}{\log c} \cdot \frac{0.31}{2} \right\} \le 0.31\\
	\impliedby& k = \frac{0.31}{1 + \frac{\log t}{\log c} \cdot \frac{0.31}{2}} = \frac{0.31}{1 + \frac{0.31}{2} \cdot \eps}.
	\end{align*}
	
	Note since $\eps \in [0,1]$, $k$ is indeed $\Theta(1)$.
	
	Finally, with our choice of $k$ specified, our running time is $n^{2 + o(1)} / b^2 = n^{2 +o(1)} / t^r$.
	
	By a simple calculation,
	\begin{align*}
	\log t^r &= r \cdot \log t \\
	&= k \cdot \log n / \log c \cdot \log t\\
	&= \log n \cdot \left\{ \frac{\log t}{\log c} \cdot \frac{0.31}{1 + \frac{0.31}{2} \cdot \eps} \right\}\\
	&= \log n \cdot \frac{0.31\eps}{1 + \frac{0.31}{2} \cdot \eps}\\
	&= \log n \cdot \frac{0.31}{\eps^{-1} + \frac{0.31}{2}}. 
	\end{align*}
	
	Hence, our running time is
	\[
	n^{2 + o(1)} / t^r = n^{2 + o(1) - \frac{0.31}{\eps^{-1} + \frac{0.31}{2}}}
	\]
	as stated.
	
	\paragraph*{The Second Algorithm.} The second algorithm follows exactly the same except for that we apply Theorem~\ref{theo:fast-matrix-mult-polylog} instead, hence the constant $0.31$ is replaced by $0.17$.
\end{proof}

\subsubsection*{Generalization to Non-negative Real Case}

Note that Lemma~\ref{lm:simple-approx} indeed works for a set of non-negative reals, we can observe that the above algorithm in fact works for $\posMaxIP_{n,d}$ (which is the same as $\MaxIP$ except for that the sets consisting of non-negative reals):\footnote{In the following we assume a real RAM model of computation for simplicity.}

\begin{reminder}{Corollary~\ref{cor:gen-posR}}
	Assuming $\omega(\log n) < d < n^{o(1)}$ and letting $\eps = \min\left( \frac{\log t}{\log (d/\log n)}, 1 \right)$, there is a $t$-\Mapprox\ deterministic algorithm for $\posMaxIP_{n,d}$ running in time
	\[
	O\left( n^{2 - \Omega(\eps)} \cdot \polylog(n) \right).
	\]
\end{reminder}
\begin{proofsketch}
	
	We can just use the same algorithm in Lemma~\ref{lm:algo-Max-IP-M}, the only difference is on the analysis of the number of terms in $P_{r}(z)$: since $z$ is no longer Boolean, $P_r(z)$ is no longer multi-linear, and we need to switch to a general upper bound $\binom{d + r}{r}$ on the number of terms for $r$-degree polynomials of $d$ variables. This corollary then follows by a similar calculation as in Lemma~\ref{lm:algo-Max-IP-M}.
\end{proofsketch}

\subsubsection*{The Lower Bound}

The lower bound follows directly from the new $\MA$ protocol for Set-Disjointness in~\cite{Rubinstein2017closest}. We present an explicit proof here for completeness. 

Before proving the lower bound we need the following reduction from $\OV$ to $t$-\Mapprox\ $\MaxIP$.

\begin{lemma}[Implicit in Theorem~4.1 of~\cite{Rubinstein2017closest}]	\label{lm:OV-to-MaxIP}
	There is a universal constant $c_1$ such that, for every integer $c$, reals $\varepsilon \in (0,1]$ and $\tau \ge 2$, $\OV_{n,c \log n}$ can be reduced to $n^{\varepsilon}$ $\MaxIP_{n,d}$ instances $(A_i,B_i)$ for $i \in [n^{\varepsilon}]$, such that:
	
	\begin{itemize}
		\item $d = \tau^{\poly(c /\varepsilon) } \cdot \log n$.
		\item Letting $T = c \log n \cdot \tau^{c_1}$, if there is $a \in A$ and $b \in B$ such that $a \cdot b = 0$, then there exists an $i$ such that $\OPT(A_i,B_i) \ge T$.
		\item Otherwise, for all $i$ we must have $\OPT(A_i,B_i) \le T/\tau$.
	\end{itemize}
\end{lemma}

The reduction above follows directly from the new $\MA$ communication protocols in~\cite{Rubinstein2017closest} together with the use of expander graphs to reduce the amount of random coins. A proof for the lemma above can be found in Appendix~\ref{app:OV-MaxIP-reduction}. 

Now we are ready to show the lower bound on $t$-\Mapprox\ $\MaxIP$.

\begin{cor}\label{cor:lowb-Max-IP-M}
	Assuming SETH (or OVC), and letting $d = \omega(\log n)$ and $t \ge 2$. There is no $n^{2 - \Omega(1)}$-time $t$-\Mapprox\ algorithm for $\MaxIP_{n,d}$ if
	\[
	t = \left(d/\log n\right)^{o(1)}.
	\]
\end{cor}

\begin{proof}
	Let $c = d / \log n$, then $t = c^{o(1)}$ (recall that $t$ and $d$ are two functions of $n$).
	
	Suppose for contradiction that there is an $n^{2-\varepsilon'}$ time $t(n)$-\Mapprox\ algorithm $\alg$ for $\MaxIP(n,d)$ for some $\varepsilon' > 0$.
	
	Let $\varepsilon = \varepsilon'/2$. Now, for every constant $c_2$, we apply the reduction in Lemma~\ref{lm:OV-to-MaxIP} with $\tau = t$ to reduce an $\OV_{n,c_2 \log n}$ instance to $n^{\varepsilon}$ 
	$$
	\MaxIP_{n,t^{\poly(c_2 / \varepsilon)} \cdot \log n} \equiv \MaxIP_{n,t^{O(1)} \cdot \log n}
	$$
	
	instances. Since $t = c^{o(1)}$, which means for sufficiently large $n$, $t^{O(1)} \cdot \log n = c^{o(1)} \cdot \log n = o(d)$, and it in turn implies that for sufficiently large $n$, $n^{\varepsilon}$ calls to $\alg$ are enough to solve the $\OV_{n,c_2 \log n}$ instance.
	
	Therefore, we can solve $\OV_{n,c_2 \log n}$ in $n^{2 - \varepsilon'} \cdot n^{\varepsilon} = n^{2 - \varepsilon}$ time for all constant $c_2$. Contradiction to OVC.
\end{proof}

Finally, the correctness of Theorem~\ref{theo:Max-IP-M} follows directly from Lemma~\ref{lm:algo-Max-IP-M} and Corollary~\ref{cor:lowb-Max-IP-M}.

\subsection{The Additive Case}

In this subsection we prove Theorem~\ref{theo:Max-IP-A}. We first recap it here for convenience.

\begin{reminder}{Theorem~\ref{theo:Max-IP-A}}
	Letting $\omega(\log n) < d < n^{o(1)}$ and $0 \le t \le d$, the following holds:
	
	\begin{enumerate}
		\item There is an $n^{2 - \Omega(1)}$-time $t$-\Aapprox\ algorithm for $\MaxIP_{n,d}$ if
		$$
		t = \Omega(d),
		$$
		and under SETH (or OVC), there is no $n^{2 - \Omega(1)}$-time $t$-\Aapprox\ algorithm for $\MaxIP_{n,d}$ if
		$$
		t = o(d).
		$$
		\item Moreover, letting $\eps = \frac{t}{d}$, there is an
		\[
		O\left(n^{2 - \Omega(\eps^{1/3}/\log\eps^{-1})}\right)
		\]
		time, $t$-\Aapprox\ randomized algorithm for $\MaxIP_{n,d}$ when $\eps \gg \log^6\log n / \log^3 n$.
	\end{enumerate}
\end{reminder}
We proceed similarly as in the multiplicative case by establishing the algorithm first.

\subsubsection*{The Algorithm}

The algorithm is actually very easy, we simply apply the following algorithm from~\cite{alman2016polynomial}.

\begin{lemma}[Implicit in Theorem~5.1 in~\cite{alman2016polynomial}]~\label{lm:algo-previous}
	Assuming $\eps \gg \log^6\log(d \log n) / \log^3 n$, there is an \[
	n^{2 - \Omega\big(\eps^{1/3} / \log (\frac{d}{\eps \log n})\big)}
	\] time $\eps \cdot d$-\Aapprox\ randomized algorithm for $\MaxIP_{n,d}$.
\end{lemma}

\begin{lemma}~\label{lm:algo-Max-IP-A}
	Let $\eps = \frac{\min(t,d)}{d}$, there is an
	\[
	O\left(n^{2 - \Omega(\eps^{1/3}/\log\eps^{-1})}\right)
	\]
	time, $t$-\Aapprox\ randomized algorithm for $\MaxIP_{n,d}$ when $\eps \gg \log^6\log n / \log^3 n$.
\end{lemma}
\begin{proof}
	When $t > d$ the problem becomes trivial, so we can assume $t \le d$, and now $t = \eps \cdot d$.
	
	Let $\eps_1 = \eps / 2$ and $c_1$  be a constant to be specified later. Given an $\MaxIP_{n,d}$ instance with two sets $A$ and $B$ of vectors from $\{0,1\}^d$,
	we create another $\MaxIP_{n,d_1}$ instance with sets $\WT{A}$ and $\WT{B}$ and $d_1 = c_1 \cdot \eps_1^{-2} \cdot \log n$ as follows:
	
	\begin{itemize}
		\item Pick $d_1$ uniform random indices $i_1,i_2,i_3,\dotsc,i_{d_1} \in [d]$, each $i_k$ is an independent uniform random number in $[d]$.
		
		\item Then we construct $\WT{A}$ from $A$ by reducing each $a \in A$ into $\tilde{a} = (a_{i_1},a_{i_2},\dotsc,a_{i_{d_1}}) \in \{0,1\}^{d_1}$ and $\WT{B}$ from $B$ in the same way.
	\end{itemize}

	Note for each $a \in A$ and $b \in B$, by a Chernoff bound, we have
	
	\[
	\Pr\left[ \left| \frac{\tilde{a} \cdot \tilde{b}}{d_1} - \frac{a \cdot b}{d} \right| \ge \eps_1 \right] < 2e^{-2d_1\eps_1^{2}} = 2n^{-2\cdot c_1}.
	\]
	
	By setting $c_1 = 2$, the above probability is smaller than $1/n^{3}$.
	
	Hence, by a simple union bound, with probability at least $1 - 1/n$, we have
	
	\[
	\left| \frac{\WT{a} \cdot \WT{b}}{d_1} - \frac{a \cdot b}{d} \right| \le \eps_1
	\]
	
	for all $a \in A$ and $b \in B$. Hence, it means that this reduction only changes the ``relative inner product''($\frac{a\cdot b}{d}$ or $\frac{\WT{a} \cdot \WT{b}}{d_1}$) of each pair by at most $\eps_1$. Hence the maximum of the ``relative inner product'' also changes by at most $\eps_1$, and we have $| \OPT(A,B)/d - \OPT(\WT{A},\WT{B})/d_1 | \le \eps_1$.
	
	Then we apply the algorithm in Lemma~\ref{lm:algo-previous} on the instance with sets $\WT{A}$ and $\WT{B}$ with error $\eps = \eps_1$ to obtain an estimate $\WT{O}$, and our final answer is simply $\frac{\WT{O}}{d_1} \cdot d$.
	
	From the guarantee from Lemma~\ref{lm:algo-previous}, we have $|\OPT(\WT{A},\WT{B})/d_1 - \WT{O}/d_1| \le \eps_1$, and therefore we have $| \OPT(A,B)/d - \WT{O}/d_1 | \le 2 \eps_1 = \eps$, from which the correctness of our algorithm follows directly.
	
	For the running time, note that the reduction part runs in linear time $O(n \cdot d)$, and the rest takes
	\[
	n^{2 - \Omega\big(\eps^{1/3} / \log (\frac{d_1}{\eps_1 \log n}) \big)} = n^{2 - \Omega(\eps^{1/3} / \log \eps^{-1})}
	\]
	time.
\end{proof}

\subsubsection*{The Lower Bound}

The lower bound is already established in~\cite{Rubinstein2017closest}, we show it follows from Lemma~\ref{lm:OV-to-MaxIP} here for completeness.

\begin{lemma}[Theorem 4.1 of~\cite{Rubinstein2017closest}]\label{lm:lowb-Max-IP-A}
	Assuming SETH (or OVC), and letting $d = \omega(\log n)$ and $t > 0$, there is no $n^{2 - \Omega(1)}$-time $t$-\Aapprox\ randomized algorithm for $\MaxIP_{n,d}$ if
	\[
	t = o(d).
	\]
\end{lemma}

\begin{proof}
		
	Recall that $t$	and $d$ are all functions of $n$. Suppose for contradiction that there is an $n^{2-\varepsilon'}$ time $t(n)$-\Aapprox\ algorithm $\alg$ for $\MaxIP(n,d)$ for some $\varepsilon' > 0$.
	
	Let $\varepsilon = \varepsilon'/2$. Now, for every constant $c_2$, we apply the reduction in Lemma~\ref{lm:OV-to-MaxIP} with $\tau = 2$ to reduce an $\OV_{n,c_2 \log n}$ instance to $n^{\varepsilon}$ 
	$$
	\MaxIP_{n,2^{\poly(c_2 / \eps)} \cdot \log n} \equiv \MaxIP_{n, d_1} \text{ where $d_1 = O(1) \cdot \log n$}
	$$
	
	instances. In addition, from Lemma~\ref{lm:OV-to-MaxIP}, to solve the $\OV_{c_2 \log n}$ instance, we only need to distinguish an additive gap of $\frac{T}{2} = \Omega(\log n) = \Omega(d_1)$ for these $\MaxIP$ instances obtained via the reduction.
	
	\newcommand{\Inew}{\mathcal{I}^{\textsf{new}}}
	
	This can be solved, via $n^{\varepsilon}$ calls to $\alg$ as follows: for each $\MaxIP_{n, d_1}$ instance $\mathcal{I}$ we get, since $d = \omega(\log n)$, which means for a sufficiently large $n$, $d_1 = O(\log n) \ll d$, and we can duplicate each coordinate $d/d_1$ times (for simplicity we assume $d_1 | d$ here), to obtain an $\MaxIP_{n, d}$ instance $\Inew$, such that $\OPT(\Inew) = d/d_1 \cdot \OPT(\mathcal{I})$. Then $\alg$ can be used to estimate $\OPT(\Inew)$ within an additive error $t = o(d)$. Scaling its estimate by $\frac{d_1}{d}$, it can also be used to estimate $\OPT(\mathcal{I})$ within an additive error $o(d_1) = o(\log n) \le T/2$ for sufficiently large $n$.
	
	Therefore, we can solve $\OV_{n,c_2 \log n}$ in $n^{2 - \varepsilon'} \cdot n^{\varepsilon} = n^{2 - \varepsilon}$ time for all constant $c_2$. Contradiction to OVC.
\end{proof}

Finally, the correctness of Theorem~\ref{theo:Max-IP-A} follows directly from Lemma~\ref{lm:algo-Max-IP-A} and Lemma~\ref{lm:lowb-Max-IP-A}.

\subsection{Adaption for $\AllPairMaxIP$}

Now we sketch the adaption for our algorithms to work for the $\AllPairMaxIP$ problem.

\begin{reminder}{Corollary~\ref{cor:All-Pair-Max-IP}}
	Suppose $ \omega(\log n) < d < n^{o(1)}$, and let 
	\[
	\eps_M := \min\left( \frac{\log t}{\log (d/\log n)}, 1 \right) \text{ and } \eps_A := \frac{\min(t,d)}{d}.
	\]
	There is an $n^{2 - \Omega(\eps_M)} \polylog(n)$ time $t$-\Mapprox\ algorithm and an $n^{2 - \Omega(\eps_A^{1/3}/\log \eps_A^{-1})}$ time $t$-\Aapprox\ algorithm for $\AllPairMaxIP_{n,d}$, when $\eps_A \gg \log^6\log n / \log^3 n$.
\end{reminder}
\begin{proofsketch}
	Note that the algorithm in Lemma~\ref{lm:algo-previous} from~\cite{alman2016polynomial} actually works for the $\AllPairMaxIP_{n,d}$. Hence, we can simply apply that algorithm after the coordinate sampling phase, and obtain a $t$-\Aapprox\ algorithm for $\AllPairMaxIP_{n,d}$.
	
	For $t$-\Mapprox\ algorithm, suppose we are given with two sets $A$ and $B$ of $n$ vectors from $\{0,1\}^{d}$. Instead of partitioning both of them into $n/b$ subsets $A_i$'s and $B_i$'s (the notations used here are the same as in the proof of Lemma~\ref{lm:algo-Max-IP-M}), we only partition $B$ into $n/b$ subsets $B_1,B_2,\dotsc,B_{n/b}$ of size $b$, and calculate $P_r(x,B_i) := \sum_{y \in B_i} P_r(x,y)$ for all $x \in A$ and $i \in [n/b]$ using similar reduction to rectangle matrix multiplication as in Lemma~\ref{lm:algo-Max-IP-M}. By a similar analysis, these can be done in $n^{2 - \Omega(\eps_M)} \cdot \polylog(n)$ time, and with these informations we can compute the $t$-\Mapprox\ answers for the given $\AllPairMaxIP_{n,d}$ instance.
\end{proofsketch}


\subsection{Improved Hardness for LCS-Closest Pair Problem}

We finish this section with the proof of Corollary~\ref{cor:LCSP}. First we abstract the reduction from $\MaxIP$ to $\LCSP$ in~\cite{ARW17-proceedings} here.

\begin{lemma}[Implicit in Theorem~1.6~in~\cite{ARW17-proceedings}]~\label{lm:MaxIP-to-LCS-Closest-Pair}
	For big enough $t$ and $n$, $t$-\Mapprox\ $\MaxIP_{n, d}$ reduces to $t/2$-\Mapprox\ $\LCSP_{n,O(d^3 \log^2 n)}$.
\end{lemma}

Now we are ready to prove Corollary~\ref{cor:LCSP} (restated below for convenience).

\begin{reminder}{Corollary~\ref{cor:LCSP}}
	Assuming SETH (or OVC), for every $t \ge 2$, $t$-\Mapprox\ $\LCSP_{n,d}$ requires $n^{2 - o(1)}$ time, if $d = t^{\omega(1)} \cdot \log^5 n$.
\end{reminder}
\begin{proof}
	From Lemma~\ref{lm:lowb-Max-IP-M}, assuming SETH (or OVC), for every $t \ge 2$, we have that $2t$-\Mapprox\ $\MaxIP_{n, d}$ requires $n^{2 - o(1)}$ time if $d = t^{\omega(1)} \cdot \log n$. Then from Lemma~\ref{lm:MaxIP-to-LCS-Closest-Pair}, we immediately have that $t$-\Mapprox\ $\LCSP_{n,d^3 \cdot \log^2 n} = \LCSP_{n,t^{\omega(1)} \cdot \log^5 n}$ requires $n^{2 - o(1)}$ time.
\end{proof}
	\section{Hardness of Exact $\IntMaxIP$, Hopcroft's Problem and More}\label{sec:exact}

In this section we show hardness of Hopcroft's problem, exact $\IntMaxIP$, $\ell_2$-Furthest Pair and Bichromatic $\ell_2$-Closest Pair. Essentially our results follow from the framework of~\cite{Wil18}, in which it is shown that hardness of Hopcroft's problem implies hardness of other three problems, and is implied by dimensionality reduction for $\OV$.
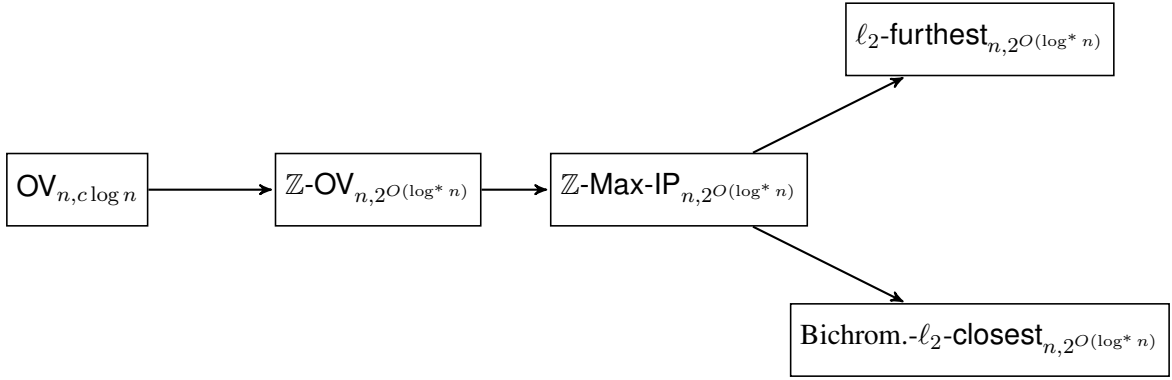
\begin{figure}[H]
	\centering
\begin{tikzpicture}[->,>=stealth',shorten >=1pt,auto,
semithick,scale = 2.0]
\tikzstyle{every state}=[draw=black,text=black,rectangle]
\tikzstyle{reduceto}=[thick]

\node [state] (OV) at (0,0) {$\OV_{n,c \log n}$};
\node [state] (Hopcroft) at (2,0) {$\Hopcroft_{n,2^{O(\logstar n)}}$};
\node [state] (IntMaxIP) at (4,0) {$\IntMaxIP_{n,2^{O(\logstar n)}}$};
\node [state] (furthest) at (6,1) {$\ell_2\text{-}\textsf{furthest}_{n,2^{O(\logstar n)}}$};
\node [state] (closest) at (6,-1) {Bichrom.-$\ell_2\text{-}\textsf{closest}_{n,2^{O(\logstar n)}}$};

\draw[reduceto] (OV) to (Hopcroft);
\draw[reduceto] (Hopcroft) to (IntMaxIP);
\draw[reduceto] (IntMaxIP) to (furthest);
\draw[reduceto] (IntMaxIP) to (closest);
\end{tikzpicture}
\caption{A diagram for all reductions in this section.}
\label{fig:reduction-graph}
\end{figure}

\subsubsection*{The Organization of this Section}
In Section~\ref{sec:improved-dim-reduction} we prove the improved dimensionality reduction for $\OV$. In Section~\ref{sec:hardness-Hopcroft} we establish the hardness of Hopcroft's problem in $2^{O(\logstar n)}$ dimensions with the improved reduction. In Section~\ref{sec:hardness-Int-Max-IP} we show Hopcroft's problem can be reduced to $\IntMaxIP$ and thus establish the hardness for the later one. In Section~\ref{sec:hardness-Geometry} we show $\IntMaxIP$ can be reduced to $\ell_2$-Furthest Pair and Bichromatic $\ell_2$-Closest Pair, therefore the hardness for the later two problems follow. See Figure~\ref{fig:reduction-graph} for a diagram of all reductions covered in this section.

The reduction in last three subsections are all from~\cite{Wil18} (either explicit or implicit), we make them explicit here for our ease of exposition and for making the paper self-contained.

\subsection{Improved Dimensionality Reduction for $\OV$}\label{sec:improved-dim-reduction}

We begin with the improved dimensionality reduction for $\OV$. The following theorem is one of the technical cores of this paper, which makes use of the CRR encoding (see Theorem~\ref{theo:CRR}) recursively. 

\begin{theo}\label{theo:main-reduction}
	Let $b, \ell$ be two sufficiently large integers. There is a reduction $\psi_{b,\ell} : \{0,1\}^{b \cdot \ell} \to \mathbb{Z}^{\ell}$ and a set $V_{b,\ell} \subseteq \mathbb{Z}$, such that for every $x,y \in \{0,1\}^{b \cdot \ell}$,
	
	$$
	x \cdot y = 0 \Leftrightarrow \psi_{b,\ell}(x) \cdot \psi_{b,\ell}(y) \in V_{b,\ell}
	$$
	and 
	$$
	0 \le \psi_{b,\ell}(x)_i < {\ell}^{6^{\logstar(b)} \cdot b}
	$$
	for all possible $x$ and $i \in [\ell]$.
	Moreover, the computation of $\psi_{b,\ell}(x)$ takes $\operatorname*{poly}(b \cdot \ell)$ time, and the set $V_{b,\ell}$ can be constructed in $O\left(\ell^{O(6^{\log^*(b)} \cdot b)} \cdot \operatorname{poly}(b \cdot \ell) \right)$ time.
\end{theo}

\begin{rem}
	We didn't make much effort to minimize the base $6$ above to keep the calculation clean, it can be replaced by any constant $ > 2 $ with a tighter calculation.
\end{rem}

\begin{proof}
	
	We are going to construct our reduction in a recursive way. $\ell$ will be the same throughout the proof, hence in the following we use $\psi_{b}$ ($V_b$) instead of $\psi_{b,\ell}$ ($V_{b,\ell}$) for simplicity.
	
	\paragraph*{Direct CRR for small $b$:}
	
	When $b < \ell$, we use a direct Chinese remainder representation of numbers. We pick $b$ primes $q_1,q_2,\dotsc,q_{b}$ in $[\ell+1,\ell^2]$, and use them for our CRR encoding.
	
	Let $x \in \{0,1\}^{b \cdot \ell}$, we partition it into $\ell$ equal size groups, and use $x^i$ to denote the  $i$-th group, which is the sub-vector of $x$ from the $((i-1) \cdot b + 1)$-th bit to the $(i \cdot b)$-th bit.
	
	Then we define $\psi_b(x)$ as
	
	$$
	\psi_b(x) := \left(\CRR\left(\left\{x^1_j\right\}_{j=1}^{b}\right),\CRR\left(\left\{x^2_j\right\}_{j=1}^{b}\right),\dotsc,\CRR\left(\left\{x^{\ell}_j\right\}_{j=1}^{b}\right) \right).
	$$
	
	That is, the $i$-th coordinate of $\psi_b(x)$ is the CRR encoding of the $i$-th sub-vector $x^i$ with respect to the primes $q_j$'s.
	
	Now, for $x,y \in \{0,1\}^{b \cdot \ell}$, note that for $j \in [b]$,
	\begin{align*}
	&\psi_b(x) \cdot \psi_b(y) \pmod{q_j}\\
	\equiv&\sum_{i=1}^{\ell} \CRR\left(\left\{x^i_j\right\}_{j=1}^{b}\right) \cdot \CRR\left(\left\{y^i_j\right\}_{j=1}^{b}\right) \pmod{q_j} \\
	\equiv&\sum_{i=1}^{\ell} x^i_j \cdot y^i_j \pmod{q_j}.
	\end{align*}
	
	Since the sum $\sum_{i=1}^{\ell} x^i_j \cdot y^i_j$ is in $[0, \ell]$, and $q_j > \ell$, we can see
	$$
	\sum_{i=1}^{\ell} x^i_j \cdot y^i_j = 0 \Leftrightarrow \psi_b(x) \cdot \psi_b(y) \equiv 0 \pmod{q_j}.
	$$
	
	Therefore, $x \cdot y = \sum_{j=1}^{b} \sum_{i=1}^{\ell} x^i_j \cdot y^i_j = 0$ is equivalent to that
	$$
	\psi_b(x) \cdot \psi_b(y) \equiv 0 \pmod{q_j}
	$$
	for every $j \in [b]$.
	
	Finally, we have $0 \le \psi_b(x)_i < \prod_{j=1}^{b} p_j < \ell^{2 \cdot b} \le \ell^{6^{\logstar(b)} \cdot b}$. Therefore 
	$$
	\psi_b(x) \cdot \psi_b(y) < \ell^{6^{\logstar(b)} \cdot 2b  + 1},
	$$
	and we can set $V_b$ to be the set of all integers in $[0,\ell^{6^{\logstar(b)} \cdot 2b  + 1}]$ that is $0$ modulo all the $p_j$'s, and it is easy to see that
	
	$$
	x \cdot y \Leftrightarrow \psi_b(x) \cdot \psi_b(y) \in V_b
	$$
	
	for all $x,y \in \{0,1\}^{b \cdot \ell}$.
	
	\paragraph*{Recursive Construction for larger $b$:} 
	
	When $b \ge \ell$, suppose the theorem holds for all $b' < b$. Let $\bm$ be the number such that (we ignore the rounding issue here and pretend that $\bm$ is an integer for simplicity),
	$$
	{\ell}^{6^{\logstar(\bm)} \cdot \bm} = b.
	$$
	
	Then we pick $b/\bm$ primes $p_1,p_2,\dotsc,p_{b/\bm}$ in $[(b^2\ell),(b^2\ell)^2]$, and use them as our reference primes in the CRR encodings.
	
	Let $x \in \{0,1\}^{b \cdot \ell}$, as before, we partition $x$ into $\ell$ equal size sub-vectors $x^1,x^2,\dotsc,x^{\ell}$, where $x^i$ consists of the $((i-1) \cdot b + 1)$-th bit of $x$ to the $(i \cdot b)$-th bit of $x$. Then we partition each $x^i$ again into $b/\bm$ micro groups, each of size $\bm$.  We use $x^{i,j}$ to denote the $j$-th micro group of $x^i$ after the partition.
	
	Now, we use $x^{[j]}$ to denote the concatenation of the vectors $x^{1,j},x^{2,j},\dotsc,x^{\ell,j}$. That is, $x^{[j]}$ is the concatenation of the $j$-th micro group in each of the $\ell$ groups.  Note that $x^{[j]} \in \{0,1\}^{\bm \cdot \ell}$, and can be seen as a smaller instance, on which we can apply $\psi_{\bm}$.
	
	Our recursive construction then goes in two steps. In the first step, we make use of $\psi_{\bm}$, and transform each $\bm$-size micro group into a single number in $[0,b)$. This step transforms $x$ from a vector in $\{0,1\}^{b \cdot \ell}$ into a vector $S(x)$ in $\mathbb{Z}^{ (b / \bm) \cdot \ell }$. And in the second step, we use a similar CRR encoding as in the base case to encode $S(x)$, to get our final reduced vector in $\mathbb{Z}^{\ell}$.
	
	$S(x)$ is simply 
	\begin{align*}
	S(x) := \Big(
	&\psi_{\bm}(x^{[1]})_1,\psi_{\bm}(x^{[2]})_1,\dotsc,\psi_{\bm}(x^{[b/\bm]})_1,\\
	&\psi_{\bm}(x^{[1]})_2,\psi_{\bm}(x^{[2]})_2,\dotsc,\psi_{\bm}(x^{[b/\bm]})_2,\\
	&\dotsc,\dotsc,\dotsc\\
	&\psi_{\bm}(x^{[1]})_{\ell},\psi_{\bm}(x^{[2]})_{\ell},\dotsc,\psi_{\bm}(x^{[b/\bm]})_{\ell}
	\Big).
	\end{align*}
	
	That is, we apply $\psi_{\bm}$ on all the $x^{[j]}$'s, and shrink all the corresponding micro-groups in $x$ into integers. Again, we partition $S$ into $\ell$ equal size groups $S^1,S^{2},\dotsc,S^{\ell}$.
	
	Then we define $\psi_b(x)$ as
	
	$$
	\psi_b(x) := \left(\CRR\left(\left\{S^1_j\right\}_{j=1}^{b/\bm}\right),\CRR\left( \left\{ S^2_j \right\}_{j=1}^{b/\bm}\right),\dotsc,\CRR\left(\left\{ S^{\ell}_j \right\}_{j=1}^{b/\bm}\right) \right).
	$$
	
	In other words, the $i$-th coordinate of $\psi_b(x)$ is the CRR representation of the number sequence $S^{i}$, with respect to our primes $\{q_j\}_{j=1}^{b/\bm}$.
	
	Now, note that for $x,y \in \{0,1\}^{b \cdot \ell}$, $x \cdot y = 0$ is equivalent to $x^{[j]} \cdot y^{[j]} = 0$ for every $j \in [b/\bm]$, which is further equivalent to 
	$$
	\psi_{\bm}(x^{[j]}) \cdot \psi_{\bm}(y^{[j]}) \in V_{\bm} 
	$$
	for all $j \in [b/\bm]$, by our assumption on $\psi_{\bm}$.
	
	Since $ 0 \le \psi_{\bm}(x^{[j]})_i, \psi_{\bm}(y^{[j]})_i < b$ for all $x,y \in \{0,1\}^{b \cdot \ell}$, $i \in [\ell]$ and $j \in [b/\bm]$, we also have $\psi_{\bm}(x^{[j]}) \cdot \psi_{\bm}(y^{[j]}) < b^2 \cdot \ell$, therefore we can assume that $V_{\bm} \subseteq [0,b^2 \ell)$.
	
	For all $x,y \in \{0,1\}^{b \cdot \ell}$ and $j \in [b/\bm]$, we have
	
	\begin{align*}
	&\psi_b(x) \cdot \psi_b(y) \\
	\equiv&\sum_{i=1}^{\ell} \CRR\left(\left\{S(x)^i_j\right\}_{j=1}^{b/\bm}\right) \cdot \CRR\left(\left\{S(y)^i_j\right\}_{j=1}^{b/\bm}\right)  \pmod{p_j}\\
	\equiv&\sum_{i=1}^{\ell} S(x)^i_j \cdot S(y)^i_j \pmod{p_j}\\
	\equiv&\sum_{i=1}^{\ell} \psi_{\bm}(x^{[j]})_i \cdot \psi_{\bm}(y^{[j]})_i \pmod{p_j}\\
	\equiv& \psi_{\bm}(x^{[j]}) \cdot \psi_{\bm}(y^{[j]}) \pmod{p_j}.
	\end{align*}
	
	Since $p_j \ge b^2 \cdot \ell$, we can determine $\psi_{\bm}(x^{[j]}) \cdot \psi_{\bm}(y^{[j]})$ from $\psi_b(x) \cdot \psi_b(y)$ by taking modulo $p_j$. Therefore, 
	
	$$
	x \cdot y = 0
	$$
	is equivalent to
	
	$$
	\left( \psi_b(x) \cdot \psi_b(y)~~\mathrm{mod}~~p_j \right) \in V_{\bm},
	$$
	for every $j \in [b/\bm]$.
	
	Finally, recall that we have
	$$
	{\ell}^{6^{\logstar(\bm)} \cdot \bm} = b.
	$$
	
	Taking logarithm of both sides, we have
	$$
	6^{\logstar(\bm)} \cdot \bm \cdot \log \ell = \log b.
	$$
	
	Then we can upper bound $\psi_b(x)_i$ by
	\begin{align*}
	\psi_b(x)_i &< \prod_{j=1}^{b/\bm} p_j \\
	&< (b^2 \ell)^{2 \cdot (b/\bm)} \tag{$b \ge \ell.$}\\
	&\le 2^{6 \cdot b / \bm \cdot \log b} \\
	&\le 2^{6 \cdot b / \bm \cdot 6^{\logstar(\bm)} \cdot \bm \cdot \log \ell} \\
	&\le \ell^{6 \cdot 6^{\logstar(\bm)} \cdot b }\\
	&\le \ell^{6^{\logstar(b)} \cdot b } \tag{$\bm \le \log b, \logstar(\bm) + 1 \le \logstar(\log b) + 1 = \logstar(b)$.}\\
	\end{align*}
	
	Therefore, we can set $V_b$ as the set of integer $t$ in $[0, \ell^{6^{\logstar(b)} \cdot 2b + 1})$ such that 
	$$
	\left( t~~\mathrm{mod}~~p_j \right) \in V_{\bm}
	$$
	for every $j \in [b/\bm]$. And it is easy to see this $V_b$ satisfies our requirement.
	
	Finally, it is easy to see that the straightforward way of constructing $\psi_b(x)$ takes $O(\operatorname*{poly}(b \cdot \ell))$ time,  and we can construct $V_{b}$ by enumerating all possible values of $\psi_b(x) \cdot \psi_b(y)$ and check each of them in $O(\operatorname*{poly}(b \cdot \ell))$ time. Since there are at most  $\ell^{O(6^{\log^*(b)} \cdot b)}$ such values, $V_b$ can be constructed in
	
	$$
	O\left(\ell^{O(6^{\log^*(b)} \cdot b)} \cdot \operatorname{poly}(b \cdot \ell) \right)
	$$
	time, which completes the proof.
	
\end{proof}

Now we prove Lemma~\ref{lm:dim-reduction-OV}, we recap its statement here for convenience.

\begin{reminder}{Lemma~\ref{lm:dim-reduction-OV}}
	Let $1 \le \ell \le d$. There is an 
	$$
	O\left(n \cdot \ell^{O(6^{\log^*d} \cdot (d/\ell))} \cdot \operatorname*{poly}(d) \right)\text{-time}
	$$
	reduction from $\OV_{n,d}$ to $\ell^{O(6^{\log^*d} \cdot (d/\ell))}$ instances of $\Hopcroft_{n,\ell + 1}$, with vectors of entries with bit-length $O\left(d/\ell \cdot \log \ell \cdot 6^{\log^* d}\right)$.
\end{reminder}
\begin{proof}
	The proof is exactly the same as the proof for Lemma 1.1 in~\cite{Wil18} with different parameters, we recap it here for convenience.
	
	Given two sets $A'$ and $B'$ of $n$ vectors from $\{0,1\}^{d}$, we apply $\psi_{d/\ell,\ell}$ to each of the vectors in $A'$ ($B'$) to obtain a set $A$ ($B$) of vectors from $\mathbb{Z}^{\ell}$.   From Theorem~\ref{theo:main-reduction}, there is a $(u,v) \in A' \times B'$ such that $u \cdot v = 0$ if and only if there is a $(u,v) \in A \times B$ such that $u \cdot v \in V_{d/\ell,\ell}$.
	
	Now, for each element $t \in V_{d/\ell,\ell}$, we are going to construct two sets $A_t$ and $B_t$ of vectors from $\mathbb{Z}^{\ell + 1}$ such that there is a $(u,v) \in A \times B$ with $u \cdot v = t$ if and only if there is a $(u,v) \in A_t \times B_t$ with $u \cdot v = 0$. We construct a set $A_t$ as a collection of all vectors $u_A = [u,1]$ for $u \in A$, and a set $B_t$ as a collection of all vectors $v_B = [v,-t]$ for $v \in B$. It is easy to verify this reduction has the properties we want.
	
	Note that there are at most $\ell^{O(6^{\log^*d} \cdot (d/\ell))}$ numbers in $V_{d/\ell,\ell}$, so we have such a number of $\Hopcroft_{n,\ell+1}$ instances. And from Theorem~\ref{theo:main-reduction}, the reduction takes
	$$
	O\left(n \cdot \ell^{O(6^{\log^*d} \cdot (d/\ell))} \cdot \operatorname*{poly}(d) \right)
	$$
	time.
	
	Finally, the bit-length of reduced vectors is bounded by
	$$
	\log\left(\ell^{O(6^{\log^*d} \cdot (d/\ell))}\right) = O\left(d/\ell \cdot \log \ell \cdot 6^{\log^*d} \right),
	$$
	which completes the proof.        
\end{proof}

\subsubsection*{A Transformation from Nonuniform Construction to Uniform Construction}

The proof for Theorem~\ref{theo:main-reduction} works recursively. In one recursive step, we reduce the construction of $\psi_{b,\ell}$ to the construction of $\psi_{\bm,\ell}$, where $\bm \le \log b$. Applying this reduction $\logstar n$ times, we get a sufficiently small instance that we can switch to a direct CRR construction. 

An interesting observation here is that after applying the reduction only thrice, the block length parameter becomes $b' \le \log\log\log b$, which is so small that we can actually use brute force to find the ``optimal'' construction $\psi_{b',\ell}$ in $b^{o(1)}$ time instead of recursing deeper. Hence, to find a construction better than Theorem~\ref{theo:main-reduction}, we only need to prove the existence of such a construction. See Appendix~\ref{app:nonuniform-obs} for details.

\subsection{Improved Hardness for Hopcroft's Problem}\label{sec:hardness-Hopcroft}
In this subsection we are going to prove Theorem~\ref{theo:Hopcroft} using our new dimensionality reduction Lemma~\ref{lm:dim-reduction-OV}, we recap its statement here for completeness. 

\begin{reminder}{Theorem~\ref{theo:Hopcroft}}[Hardness of Hopcroft's Problem in $c^{\logstar n}$ Dimension]
	Assuming SETH (or OVC), there is a constant $c$ such that $\Hopcroft_{n,c^{\logstar n}}$ with vectors of $O(\log n)$-bit entries requires $n^{2-o(1)}$ time.
\end{reminder}

\begin{proof}
	The proof here follows roughly the same as the proof for Theorem~1.1 in~\cite{Wil18}.
	
	Let $c$  be an arbitrary constant and $d := c \cdot \log n$. We show that an oracle solving $\Hopcroft_{n,\ell+1}$ where $\ell = 7^{\log^* n}$ in $O(n^{2 - \delta})$ time for some $\delta > 0$ can be used to construct an $O(n^{2-\delta + o(1)})$ time algorithm for $\OV_{n,d}$, and therefore contradicts the OVC.
	
	We simply invoke Lemma~\ref{lm:dim-reduction-OV}, note that we have
	\begin{align*}
	\log \left\{ \ell^{O\left(  6^{\log^* d} \cdot (d/\ell) \right)} \right\} 
	&= \log \ell \cdot O\left(  6^{\log^* d} \cdot (d/\ell) \right)\\
	&= O\left(  \log^* n \cdot  6^{\log^* n} \cdot c \cdot \log n / 7^{\log^*n} \right)\\
	&= O\left(   \log^* n \cdot (6/7)^{\log^* n} \cdot c \cdot \log n \right)\\
	&= o(\log n).
	\end{align*}
	Therefore, the reduction takes $O(n \cdot \ell^{O\left(  6^{\log^* d} \cdot (d/\ell) \right)} \cdot \operatorname{poly}(d) ) = n^{1 + o(1)}$ time, and an $\OV_{n,d}$ instance is reduced to $n^{o(1)}$ instances of $\Hopcroft_{n,\ell+1}$, and the reduced vectors have bit length $o(\log n)$ as calculated above.  We simply solve all these $n^{o(1)}$ instances using our oracle, and this gives us an $O(n^{2-\delta + o(1)})$ time algorithm for $\OV_{n,d}$, which completes the proof.
\end{proof}

\subsection{Hardness for $\IntMaxIP$}\label{sec:hardness-Int-Max-IP}

Now we move to hardness of exact $\IntMaxIP$.

\begin{theo}[Implicit in Theorem~1.2~\cite{Wil18}]\label{theo:IntOV-to-Max-IP}
	There is an $O(\poly(d) \cdot n)$-time algorithm which reduces a $\IntOV_{n,d}$ instance into a $\IntMaxIP_{n,d^2}$ instance.
\end{theo}

\begin{proof}
	We remark here that this reduction is implicitly used in the proof of Theorem~1.2 in~\cite{Wil18}, we abstract it here only for our exposition. 
	
	Given a $\IntOV_{n,d}$ instance with sets $A,B$. Consider the following polynomial $P(x,y)$, where $x,y \in \mathbb{Z}^{d}$.
	\[
	P(x,y) = -(x \cdot y)^2 = \sum_{i,j \in [d]} - x_i \cdot y_j.
	\]
	
	It is easy to see that whether there is a $(x,y) \in A \times B$ such that $x \cdot y = 0$ is equivalent to whether the maximum value of $P(x,y)$ is $0$.
	
	Now, for each $x \in A$ and $y \in B$, we construct $\WT{x},\WT{y} \in \mathbb{Z}^{d^2}$ such that $\WT{x}_{i} = x_{ \lfloor (i-1) /d \rfloor + 1 }$ and $\WT{y}_{i} = - y_{(i \bmod{d}) + 1}$. Then we have $\WT{x} \cdot \WT{y} = P(x,y)$. Hence, let $\WT{A}$ be the set of all these $\WT{x}$'s, and $\WT{B}$ be the set of all these $\WT{y}$'s, whether there is a $(x,y) \in A \times B$ such that $x \cdot y = 0$ is equivalent to whether $\OPT(\WT{A},\WT{B}) = 0$, and our reduction is completed.
	
\end{proof}

Now, Theorem~\ref{theo:hard-Int-Max-IP} (restated below) is just a simple corollary of Theorem~\ref{theo:IntOV-to-Max-IP} and Theorem~\ref{theo:Hopcroft}.

\begin{reminder}{Theorem~\ref{theo:hard-Int-Max-IP}}
	Assuming SETH (or OVC), there is a constant $c$ such that every exact algorithm for $\IntMaxIP_{n,d}$ for $d = c^{\log^*n}$ dimensions requires $n^{2-o(1)}$ time, with vectors of $O(\log n)$-bit entries.
\end{reminder}

\subsubsection*{A Dimensionality Reduction for $\MaxIP$}

The reduction $\psi_{b,\ell}$ from Theorem~\ref{theo:main-reduction} actually does more: for $x,y \in \{0,1\}^{b \cdot \ell}$, from $\psi_{b,\ell}(x) \cdot \psi_{b,\ell}(y)$ we can in fact determine the inner product $x \cdot y$ itself, not only whether $x \cdot y = 0$.

Starting from this observation, together with Theorem~\ref{theo:IntOV-to-Max-IP}, we can in fact derive a similar dimensionality self reduction from $\MaxIP$ to $\IntMaxIP$, we deter its proof to Appendix~\ref{app:dim-reduction-MaxIP}.

\begin{cor}\label{cor:MaxIP-reduction}
	Let $1 \le \ell \le d$. There is an 
	$$
	O\left(n \cdot \ell^{O(6^{\log^*d} \cdot (d/\ell))} \cdot \operatorname*{poly}(d) \right)\text{-time}
	$$
	reduction from $\MaxIP_{n,d}$ to $ d \cdot \ell^{O(6^{\log^*d} \cdot (d/\ell))}$ instances of $\IntMaxIP_{n,(\ell + 1)^2}$, with vectors of entries with bit-length $O\left(d/\ell \cdot \log \ell \cdot 6^{\log^* d}\right)$.
\end{cor}

\subsection{Hardness for $\ell_2$-Furthest Pair and Bichromatic $\ell_2$-Closest Pair}\label{sec:hardness-Geometry}

We finish the whole section with the proof of hardness of $\ell_2$-Furthest Pair and Bichromatic $\ell_2$-Closest Pair. The two reductions below are slight adaptations of the ones in the proofs of Theorem~1.2 and Corollary~2.1 in~\cite{Wil18}.

\begin{lemma}\label{lm:Max-IP-to-furtherest-pair}
	Assuming $d = n^{o(1)}$, there is an $O(\poly(d) \cdot n)$-time algorithm which reduces a $\IntMaxIP_{n,d}$ instance into an instance of $\ell_2$-Furthest Pair on $2n$ points in $\R^{d + 2}$. Moreover, if the $\IntMaxIP$ instance consists of vectors of $O(\log n)$-bit entries, so does the $\ell_2$-Furthest Pair instance.
\end{lemma}

\begin{proof}
	Let $A,B$ be the sets in the $\IntMaxIP_{n,d}$ instance, and $k$ be the smallest integer such that all vectors from $A$ and $B$ consist of $(k \cdot \log n)$-bit entries.
	
	Let $W$ be $n^{C \cdot k}$ where $C$ is a large enough constant. Given $x \in A$ and $y \in B$, we construct point 
	\[
	\WT{x} = \left(x, \sqrt{W - \|x\|^2} ,0\right) \quad\text{and}\quad \WT{y} = \left(-y,0,\sqrt{W - \|y\|^2} \right),
	\]
	that is, appending two corresponding values into the end of vectors $x$ and $-y$.
	
	Now, we can see that for $x_1,x_2 \in A$, the squared distance between their reduced points is
	\[
	\| \WT{x_1} - \WT{x_2}  \|^2 = \| x_1 - x_2 \|^2 \le 4 \cdot d \cdot n^{2k}.
	\]
	Similarly we have
	\[
	\| \WT{y_1} - \WT{y_2}  \|^2 \le 4 \cdot d \cdot n^{2k}
	\]
	for $y_1,y_2 \in B$.
	
	Next, for $x \in A$ and $y \in B$, we have
	\[
	\| \WT{x} - \WT{y} \|^2 = \|\WT{x}\|^2 + \|\WT{y}\|^2 - 2 \cdot \WT{x} \cdot \WT{y} = 2 \cdot W + 2 \cdot (x \cdot y) \ge 2 \cdot W - d \cdot n^{2k} \gg 4 \cdot d \cdot n^{2k},
	\]
	
	the last inequality holds when we set $C$ to be $5$.
	
	Putting everything together, we can see the $\ell_2$-furthest pair among all points $\WT{x}$'s and $\WT{y}$'s must be a pair of $\WT{x}$ and $\WT{y}$ with $x \in A$ and $y \in B$. And maximizing $\|\WT{x} - \WT{y}\|$ is equivalent to maximize $x \cdot y$, which proves the correctness of our reduction. Furthermore, when $k$ is a constant, the reduced instance clearly only needs vectors with $O(k) \cdot \log n = O(\log n)$-bit entries.
\end{proof}

\begin{lemma}\label{lm:Max-IP-to-bichromatic-closest-pair}
	Assuming $d = n^{o(1)}$, there is an $O(\poly(d) \cdot n)$-time algorithm which reduces a $\IntMaxIP_{n,d}$ instance into an instance of Bichromatic $\ell_2$-Closest Pair on $2n$ points in $\R^{d + 2}$. Moreover, if the $\IntMaxIP$ instance consists of vectors of $O(\log n)$-bit entries, so does the Bichromatic $\ell_2$-Closest Pair instance.
\end{lemma}

\begin{proof}
	Let $A,B$ be the sets in the $\IntMaxIP_{n,d}$ instance, and $k$ be the smallest integer such that all vectors from $A$ and $B$ consist of $(k \cdot \log n)$-bit entries.
	
	Let $W$ be $n^{C \cdot k}$ where $C$ is a large enough constant. Given $x \in A$ and $y \in B$, we construct point 
	\[
	\WT{x} = \left(x, \sqrt{W - \|x\|^2} ,0\right) \quad\text{and}\quad \WT{y} = \left(y,0,\sqrt{W - \|y\|^2} \right),
	\]
	that is, appending two corresponding values into the end of vectors $x$ and $-y$. And our reduced instance is to find the closest point between the set $\WT{A}$ (consisting of all these $\WT{x}$ where $x \in A$) and the set $\WT{B}$ (consisting of all these $\WT{y}$ where $y \in B$).
	
	Next, for $x \in A$ and $y \in B$, we have
	\[
	\| \WT{x} - \WT{y} \|^2 = \|\WT{x}\|^2 + \|\WT{y}\|^2 - 2 \cdot \WT{x} \cdot \WT{y} = 2 \cdot W - 2 \cdot (x \cdot y) \ge 2 \cdot W - d \cdot n^{2k} \gg 4 \cdot d \cdot n^{2k},
	\]
	
	the last inequality holds when we set $C$ to be $5$.
	
	Hence minimizing $\|\WT{x} - \WT{y}\|$ where $x \in A$ and $y \in B$ is equivalent to maximize $x \cdot y$, which proves the correctness of our reduction. Furthermore, when $k$ is a constant, the reduced instance clearly only needs vectors with $O(k) \cdot \log n = O(\log n)$-bit entries.
\end{proof}

Now Theorem~\ref{theo:l-2-furthest-pair} and Theorem~\ref{theo:bi-closest-pair} (restated below) are simple corollaries of Lemma~\ref{lm:Max-IP-to-furtherest-pair}, Lemma~\ref{lm:Max-IP-to-bichromatic-closest-pair} and Theorem~\ref{theo:hard-Int-Max-IP}.

\begin{reminder}{Theorem~\ref{theo:l-2-furthest-pair}}[Hardness of $\ell_2$-Furthest Pair in $c^{\log^* n}$ Dimension]
	Assuming SETH (or OVC), there is a constant $c$ such that $\ell_2$-Furthest Pair in $c^{\log^*n}$ dimensions requires $n^{2-o(1)}$ time, with vectors of $O(\log n)$-bit entries.
\end{reminder}

\begin{reminder}{Theorem~\ref{theo:bi-closest-pair}}[Hardness of Bichromatic $\ell_2$-closest Pair in $c^{\log^* n}$ Dimension] 
	Assuming SETH (or OVC), there is a constant $c$ such that Bichromatic $\ell_2$-Closest Pair in $c^{\log^*n}$ dimensions requires $n^{2-o(1)}$ time, with vectors of $O(\log n)$-bit entries.
\end{reminder}

	\section{$\NP \cdot \UPP$ communication protocol and Exact Hardness for $\IntMaxIP$}\label{sec:NP-UPP}

 We note that the inapproximability results for (Boolean) $\MaxIP$ is established via a connection to the $\MA$ communication complexity protocol of Set-Disjointness~\cite{ARW17-proceedings}. In the light of this, in this section we view our reduction from $\OV$ to $\IntMaxIP$ (Lemma~\ref{lm:dim-reduction-OV} and Theorem~\ref{theo:IntOV-to-Max-IP}) in the perspective of communication complexity. 
 
 We observe that in fact, our reduction can be understood as an $\NP \cdot \UPP$ communication protocol for Set Disjointness. Moreover, we show that if we can get a slightly better $\NP \cdot \UPP$ communication protocol for Set-Disjointness, then we would be able to prove $\IntMaxIP$ is hard even for $\omega(1)$ dimensions (and also $\ell_2$-Furthest Pair and Bichromatic $\ell_2$-Closest Pair). 

\subsection{$\NP \cdot \UPP$ Communication Protocol for Set-Disjointness}

\newcommand{\psiAlice}{\psi_{\textsf{Alice}}}
\newcommand{\psiBob}{\psi_{\textsf{Bob}}}

First, we rephrase the results of Lemma~\ref{lm:dim-reduction-OV} and Theorem~\ref{theo:IntOV-to-Max-IP} in a more convenience way for our use here.

\begin{lemma}[Rephrasing of Lemma~\ref{lm:dim-reduction-OV} and Theorem~\ref{theo:IntOV-to-Max-IP}]\label{lm:help1}
	
	Let $1 \le \ell \le d$, and $m = \ell^{O(6^{\log^*d} \cdot (d/\ell))}$. There exists a family of functions \[
	\psiAlice^i,\psiBob^i : \{0,1\}^{d} \to \mathbb{R}^{(\ell+1)^2}
	\]
	for $i \in [m]$ such that:
	
	\begin{itemize}
		\item when $x \cdot y = 0$, there is an $i$ such that $\psiAlice^i(x) \cdot \psiBob^i(y) \ge 0$;
		\item when $x \cdot y > 0$, for all $i$ $\psiAlice^i(x) \cdot \psiBob^i(y) < 0$;
		\item all $\psiAlice^i(x)$ and $\psiBob^i(y)$ can be computed in $\poly(d)$ time.
	\end{itemize}
\end{lemma}

From the above lemma, and the standard connection between $\UPP$ and sign-rank~\cite{paturi1986probabilistic} (see also Chapter 4.11 of~\cite{jukna2012boolean}), we immediately get the communication protocol we want and prove Theorem~\ref{theo:NPUPP-for-DISJ} (restated below for convenience).

\begin{reminder}{Theorem~\ref{theo:NPUPP-for-DISJ}}
	For all $1 \le \alpha \le n$, there is an
	\[
	\left( \alpha \cdot 6^{\logstar n} \cdot (n/2^\alpha), O(\alpha) \right)\text{-computational-efficient}
	\]
	$\NP \cdot \UPP$ communication protocol for $\DISJ_{n}$.
\end{reminder}
\begin{proofsketch}
	We set $\alpha = \log \ell$ here. Given the function families $\{\psiAlice^i\},\{\psiBob^i\}$ from Lemma~\ref{lm:help1}, Merlin just sends the index $i \in [m]$, the rest follows from the connection between $\UPP$ protocols and sign-rank of matrices.
\end{proofsketch}

\subsection{Slightly Better Protocols Imply Hardness in $\omega(1)$ Dimensions}

Finally, we show that if we have a slightly better $\NP \cdot \UPP$ protocol for Set-Disjointness, then we can show $\IntMaxIP$ requires $n^{2 - o(1)}$ time even for $\omega(1)$ dimensions (and so do $\ell_2$-Furthest Pair and Bichromatic $\ell_2$-Closest Pair). We restate Theorem~\ref{theo:better-imply-MaxIP} here for convenience.

\begin{reminder}{Theorem~\ref{theo:better-imply-MaxIP}}
	Assuming SETH (or OVC), if there is an increasing and unbounded function $f$ such that for all $1 \le \alpha \le n$, there is a
	\[
	\left( n / f(\alpha), \alpha \right)\text{-computational-efficient}
	\]
	$\NP \cdot \UPP$ communication protocol for $\DISJ_{n}$, then $\IntMaxIP_{n,\omega(1)}$ requires $n^{2 - o(1)}$ time with vectors of $\polylog(n)$-bit entries. The same holds for $\ell_2$-Furthest Pair and Bichromatic $\ell_2$-Closest Pair.
\end{reminder}
\begin{proof}
	Suppose otherwise, there is an algorithm $\alg$ for $\IntMaxIP_{n,d}$ running in $n^{2 - \eps_1}$ time for all constant $d$ and for a constant $\eps_1 > 0$ (note for the sake of Lemma~\ref{lm:Max-IP-to-furtherest-pair} and Lemma~\ref{lm:Max-IP-to-bichromatic-closest-pair}, we only need to consider $\IntMaxIP$ here).
	
	Now, let $c$ be an arbitrary constant, we are going to construct an algorithm for $\OV_{n,c \log n}$ in $n^{2 - \Omega(1)}$ time, which contradicts OVC.
	
	Let $\eps = \eps_1 / 2$, and $\alpha$ be the first number such that $c/f(\alpha) < \eps$, note that $\alpha$ is also a constant. Consider the $(c \log n / f(\alpha),\alpha)$-computational-efficient $\NP \cdot \UPP$ protocol $\Pi$ for $\DISJ_{c \log n}$, and let $A,B$ be the two sets in the $\OV_{n,c \log n}$ instance. Our algorithm via reduction works as follows:
	
	\begin{itemize}
		\item There are $2^{\alpha}$ possible messages in $\{0,1\}^{\alpha}$, let $m_{1},m_{2},\dotsc,m_{2^\alpha}$ be an enumeration of them.
		
		\item We first enumerate all possible advice strings from Merlin in $\Pi$, there are $2^{c \log n / f(\alpha)} \le 2^{\eps \cdot \log n} = n^{\eps}$ such strings, let $\phi \in \{0,1\}^{\eps \cdot \log n}$ be such an advice string.
			\begin{itemize}
				\item For each $x \in A$, let $\psiAlice(x) \in \mathbb{R}^{2^{\alpha}}$ be the probabilities that Alice accepts each message from Bob. That is, $\psiAlice(x)_{i}$ is the probability that Alice accepts the message $m_i$, given its input $x$ and the advice $\phi$.
				
				\item Similarly, for each $y \in B$, let $\psiBob(y) \in \mathbb{R}^{2^\alpha}$ be the probabilities that Bob sends each message. That is, $\psiBob(y)_i$ is the probability that Bob sends the message $m_i$, give its input $y$ and the advice $\phi$.
				
				\item Then, for each $x \in A$ and $y \in B$, $\psiAlice(x) \cdot \psiBob(y)$ is precisely the probability that Alice accepts at the end when Alice and Bob holds $x$ and $y$ correspondingly and the advice is $\phi$. Now we let $A_\phi$ be the set of all the $\psiAlice(x)$'s, and $B_\phi$ be the set of all the $\psiBob(y)$'s.
			\end{itemize}
		
		\item If there is a $\phi$ such that $\OPT(A_\phi,B_\phi) \ge 1/2$, then we output yes, and otherwise output no.
	\end{itemize}

	From the definition of $\Pi$, it is straightforward to see that the above algorithm solves $\OV_{n,c \cdot \log n}$. Moreover, notice that from the computational-efficient property of $\Pi$, the reduction itself works in $n^{1+\eps} \cdot \polylog(n)$ time, and all the vectors in $A_\phi$'s and $B_\phi$'s have at most $\polylog(n)$ bit precision, which means $\OPT(A_\phi,B_\phi)$ can be solved by a call to $\IntMaxIP_{n,2^\alpha}$ with vectors of $\polylog(n)$-bit entries.
	
	Hence, the final running time for the above algorithm is bounded by $n^{\eps} \cdot n^{2 - \eps_1} = n^{2 - \eps}$ ($2^{\alpha}$ is still a constant), which contradicts the OVC.
\end{proof}
	\newcommand{\IP}{\mathsf{IP}}

\section{Improved $\MA$ Protocols}
\label{sec:MA}

In this section we prove Theorem~\ref{theo:improved-MA} (restated below for convenience).

\begin{reminder}{Theorem~\ref{theo:improved-MA}}
	There is an $\MA$ protocol for $\DISJ_{n}$ and $\InProd_{n}$ with communication complexity
	$$
	O\left(\sqrt{n\log n\log\log n}\right).
	$$
\end{reminder}

To prove Theorem~\ref{theo:improved-MA}, we need the following intermediate problem.
\begin{defi}[The Inner Product Modulo $p$ Problem ($\IP^{p}_n$)]
	Let $p$ and $n$ be two positive integers, in $\IP^{p}_{n}$, Alice and Bob are given two vectors $X$ and $Y$ in $\{0,1\}^{n}$, and they want to compute $X \cdot Y \pmod p$. 
\end{defi}

Note that $\IP_{n}$ and $\IP^p_n$ are not Boolean functions, so we need to generalize the definition of an $\MA$ protocol. In an $\MA$ protocol for $\IP_{n}$, Merlin sends the answer directly to Alice together with a proof to convince Alice and Bob. The correctness condition becomes that for the right answer $X \cdot Y$, Merlin has a proof such that Alice and Bob will accept with high probability (like $2/3$). And the soundness condition becomes that for the wrong answers, every proof from Merlin will be rejected with high probability.

We are going to use the following $\MA$ protocol for $\IP^{p}_{n}$, which is a slight adaption from the protocol in~\cite{Rubinstein2017closest}.

\begin{lemma}[Implicit in Theorem 3.1 of~\cite{Rubinstein2017closest}]\label{lm:previous}
	For a sufficiently large prime $q$ and integers $T$ and $n$, there is an 
	\[
	\Big(O\left( n/T \cdot \log q \right),
	\log n + O(1),
	O\left( T \cdot \log q  \right),
	1/2
	\Big)\text{-efficient}
	\] 
	$\MA$ protocol for $\IP^{q}_{n}$.
\end{lemma}
\begin{proofsketch}
	The only adaption is that we just use the field $\mathbb{F}_{q^2}$ with respect to the given prime $q$. (In the original protocol it is required that $q \ge T$.)
\end{proofsketch}

Now we ready to prove Theorem~\ref{theo:improved-MA}.

\begin{proofof}{Theorem~\ref{theo:improved-MA}}
	Since a $\IP_{n}$ protocol trivially implies a $\DISJ_{n}$ protocol, we only need to consider $\IP_{n}$ in the following.
	
	Now, let $x$ be the number such that $x^x = n$, for convenience we are going to pretend that $x$ is an integer. It is easy to see that $x = \Theta(\log n/\log\log n)$.
	Then we pick $10 x$ distinct primes $p_1,p_2,\dotsc,p_{10x}$ in $[x+1,x^2]$ (we can assume that $n$ is large enough to make $x$ satisfy the requirement of Lemma~\ref{lm:many-primes}).	Let $T$ be a parameter, we use $\Pi_{p_i}$ to denote the $\Big(O\left( n/T \cdot \log p_i \right),	\log n + O(1),	O\left( T \cdot \log p_i  \right), 1/2\Big)$-efficient $\MA$ protocol for $\IP_{n}^{p_i}$.
	
	Our protocol for $\IP_{n}$ works as follows:	
	\begin{itemize}
		\item Merlin sends Alice all the advice strings from the protocols $\Pi_{p_1},\Pi_{p_2},\dotsc,\Pi_{p_{10x}}$, together with a presumed inner product $0 \le z \le n$.
		
		\item Note that $\Pi_{p_i}$ contains the presumed value of $X \cdot Y \pmod{p_i}$, Alice first checks whether $z$ is consistent with all these $\Pi_{p_i}$'s, and rejects immediately if it does not.
		
		\item Alice and Bob jointly toss $O(\log (10 x))$ coins, to pick a uniform random number $i^\star \in [10x]$, and then they simulate $\Pi_{p_{i^\star}}$. That is, they pretend they are the Alice and Bob in the protocol $\Pi_{p_{i^\star}}$ with the advice from Merlin in $\Pi_{p_{i^\star}}$ (which Alice does have).
	\end{itemize}

	\paragraph*{Correctness.} Let $X,Y \in \{0,1\}^{n}$ be the vectors of Alice and Bob. If $X \cdot Y = z$, then by the definition of these protocols $\Pi_{p_i}$'s, Alice always accepts with the correct advice from Merlin. 
	
	Otherwise, let $ d = X \cdot Y \ne z$, we are going to analyze the probability that we pick a ``good'' $p_{i^\star}$ such that $p_{i^\star}$ does not divide $|d - z|$. Since $p_i > x$ for all $p_i$'s and $x^x > n \ge |d - z|$, $|d - z|$ cannot be a multiplier of more than $x$ primes in $p_i$'s. 
	
	Therefore, with probability at least $0.9$, our pick of $p_{i^\star}$ is good. And in this case, from the definition of the protocols $\Pi_{p_i}$'s, Alice and Bob would reject afterward with probability at least $1/2$, as $d \pmod{p_{i^\star}}$ differs from $z \pmod{p_{i^\star}}$. In summary, when $X \cdot Y \ne z$, Alice rejects with probability at least $0.9 / 2 = 0.45$, which finishes the proof for the correctness.
	
	\paragraph*{Complexity.} Now, note that the total advice length is
	\[
	O\left(n/T \cdot \sum_{i=1}^{10x} \log p_i\right) = 
	O\left(n/T \cdot \log \prod_{i=1}^{10x} x^2\right) =
	O\left(n/T \cdot \log x^{20x}\right) = 
	O\left(n/T \cdot \log n\right).
	\]
	And the communication complexity between Alice and Bob is bounded by
	\[
	O\left(T \cdot \log x^2 \right) =
	O\left(T \cdot \log \log n \right).
	\]
	
	Setting $T = \sqrt{n \log n / \log\log n}$ balances the above two quantities, and we obtain the needed $\MA$-protocol for $\DISJ_{n}$.
\end{proofof}

	\section{Future Works}
We end our paper by discussing a few interesting research directions.

\begin{itemize}
	\item The most important open question from this paper is that can we further improve the dimensionality reduction for $\OV$? It is certainly weird to consider $2^{O(\logstar n)}$ to be the right answer for the limit of the dimensionality reduction. This term seems more like a product of the nature of our recursive construction and not the problem itself. We conjecture that there should be an $\omega(1)$ dimensional reduction with a more direct construction. 
	
	One possible direction is to combine the original polynomial-based construction from~\cite{Wil18} together with our new number theoretical one. These two approaches seem completely different, hence a clever combination of them may solve our problem.
	
	\item In order to prove $\omega(1)$ dimensional hardness for $\ell_2$-Furthest Pair and Bichromatic $\ell_2$-Closest Pair, we can also bypass the $\OV$ dimensionality reduction things by proving $\omega(1)$ dimensional hardness for $\IntMaxIP$ directly. One possible way to approach this question is to start from the $\NP \cdot \UPP$ communication protocol connection as in Section~\ref{sec:NP-UPP} (apply Theorem~\ref{theo:better-imply-MaxIP}), and (potentially) draw some connections from some known $\UPP$ communication protocols.
	
	\item We have seen an efficient reduction from $\Hopcroft$ to $\IntMaxIP$ which only blows up the dimension quadratically, is there a similar reduction from $\IntMaxIP$ back to $\Hopcroft$? Are $\IntMaxIP$ and $\Hopcroft$ equivalent?
	
	\item By making use of the new AG-code based $\MA$ protocols, we can shave a $\widetilde{O}(\sqrt{\log n})$ factor from the communication complexity, can we obtain an $O(\sqrt{n})$ $\MA$ communication protocol matching the lower bound for $\DISJ_n$? It seems new ideas are required. 
	
	Since our $\MA$ protocol works for both $\DISJ$ and $\InProd$, and $\InProd$ does seems to be a harder problem. It may be better to find an $\MA$ protocol only works for $\DISJ$. It is worth noting that an $O(\sqrt{n})$ $\textsf{AMA}$ communication protocol for $\DISJ$ is given by~\cite{Rubinstein2017closest}, which doesn't work for $\InProd$.
	
	\item Can the dependence on $\eps$ in the algorithms from Theorem~\ref{theo:Max-IP-M} be further improved? Is it possible to apply ideas in the $n^{2 - 1/\widetilde{\Omega}(\sqrt{c})}$ algorithm for $\MaxIP_{n, c\log n}$ from~\cite{alman2016polynomial}?
	
	\item For the complexity of $2$-multiplicative-approximation to $\MaxIP_{n,c \log n}$, Theorem~\ref{theo:Max-IP-M} implies that there is an algorithm running in $n^{2 - 1/O(\log c)}$ time, the same as the best algorithm for $\OV_{n,c \log n}$~\cite{abboud2015more}. Is this just a coincidence? Or are there some connections between these two problems?
	
	\item We obtain a connection between hardness of $\IntMaxIP$ and $\NP \cdot \UPP$ communication protocols for Set-Disjointness. Can we get similar connections from other $\NP \cdot \mathcal{C}$ type communication protocols for Set-Disjointness? Some candidates include $\NP \cdot \SBP$ and $\NP \cdot \textsf{promiseBQP}$ ($\textsf{QCMA}$).
\end{itemize}

\section*{Acknowledgment}
I would like to thank Ryan Williams for introducing the problem to me, countless encouragement and helpful discussions during this work, and also many comments on a draft of this paper. In particular, the idea of improving $\OV$ dimensionality self-reduction using CRT (the direct CRT based approach) is introduced to me by Ryan Williams. 

I am grateful to Virginia Vassilevska Williams, Kaifeng Lv, Peilin Zhong for helpful discussions and suggestions. I would like to thank Aviad Rubinstein for sharing a manuscript of his paper, and pointing out that the $O(\sqrt{n\log n \log\log n})$ $\MA$ protocol also works for Inner Product.
	
	\appendix
	\section{A Dimensionality Reduction for $\MaxIP$}
\label{app:dim-reduction-MaxIP}

In fact, tracing the proof of Theorem~\ref{theo:main-reduction}, we observe that it is possible to compute the inner product $x \cdot y$ itself from $\psi_{b,\ell}(x) \cdot \psi_{b,\ell}(y)$, that is:

\begin{cor}\label{cor:ExactIP-reduction}
	Let $b, \ell$ be two sufficiently large integers. There is a reduction $\psi_{b,\ell} : \{0,1\}^{b \cdot \ell} \to \mathbb{Z}^{\ell}$ and $b \cdot \ell + 1$ sets $V_{b,\ell}^0,V_{b,\ell}^1,\dotsc,V_{b,\ell}^{b \cdot \ell} \subseteq \mathbb{Z}$, such that for every $x,y \in \{0,1\}^{b \cdot \ell}$,
	
	$$
	x \cdot y = k \Leftrightarrow \psi_{b,\ell}(x) \cdot \psi_{b,\ell}(y) \in V_{b,\ell}^k \quad \text{for all $0 \le k \le b \cdot \ell$,}
	$$
	and 
	$$
	0 \le \psi_{b,\ell}(x)_i < {\ell}^{6^{\logstar(b)} \cdot b}
	$$
	for all possible $x$ and $i \in [\ell]$.
	Moreover, the computation of $\psi_{b,\ell}(x)$ takes $\operatorname*{poly}(b \cdot \ell)$ time, and the sets $V_{b,\ell}^k$'s can be constructed in $O\left(\ell^{O(6^{\log^*(b)} \cdot b)} \cdot \operatorname{poly}(b \cdot \ell) \right)$ time.
\end{cor}

Together with Theorem~\ref{theo:IntOV-to-Max-IP}, it proves Corollary~\ref{cor:MaxIP-reduction} (restated below).

\begin{reminder}{Corollary~\ref{cor:MaxIP-reduction}}
	Let $1 \le \ell \le d$. There is an 
	$$
	O\left(n \cdot \ell^{O(6^{\log^*d} \cdot (d/\ell))} \cdot \operatorname*{poly}(d) \right)\text{-time}
	$$
	reduction from $\MaxIP_{n,d}$ to $ d \cdot \ell^{O(6^{\log^*d} \cdot (d/\ell))}$ instances of $\IntMaxIP_{n,(\ell + 1)^2}$, with vectors of entries with bit-length $O\left(d/\ell \cdot \log \ell \cdot 6^{\log^* d}\right)$.
\end{reminder}
\begin{proofsketch}
	Let $b = d / \ell$ (assume $\ell$ divides $d$ here for simplicity), $A$ and $B$ be the sets in the given $\MaxIP_{n,d}$ instance, we proceed similarly as the case for $\OV$. 
	
	We first enumerate a number $k$ from $0$ to $d$, for each $k$ we construct the set $V_{b,\ell}^k$ as specified in  Corollary~\ref{cor:ExactIP-reduction}. Then there is $(x,y) \in A \times B$ such that $x \cdot y = k$ if and only if there is $(x,y) \in A \times B$ such that $\psi_{b,\ell}(x) \cdot \psi_{b,\ell}(y) \in V_{b,\ell}^k$. Using exactly the same reduction as in Lemma~\ref{lm:dim-reduction-OV}, we can in turn reduce this into ${\ell}^{O(6^{\logstar(b)} \cdot b)}$ instances of $\Hopcroft_{n,\ell + 1}$. 
	
	Applying Theorem~\ref{theo:IntOV-to-Max-IP}, with evaluation of $(d+1) \cdot {\ell}^{O(6^{\logstar(b)} \cdot b)}$ $\IntMaxIP_{n,(\ell + 1)^2}$ instances, we can determine whether there is $(x,y) \in A \times B$ such that $x \cdot y = k$ for every $k$, from which we can compute the answer to the $\MaxIP_{n, d}$ instance.
\end{proofsketch}
	\section{Nonuniform to Uniform Transformation for Dimensionality Reduction for $\OV$}
\label{app:nonuniform-obs}

In this section we discuss the transformation from nonuniform construction to uniform one for dimensionality reduction for $\OV$. In order to state our result formally, we need to introduce some definitions.

%

\begin{defi}[Nonuniform Reduction]
	Let $b, \ell, \kappa \in \mathbb{N}$. We say a function $\varphi : \{0,1\}^{b \cdot \ell} \to \mathbb{Z}^{\ell}$ together with a set $V \subseteq \mathbb{Z}$ is a $(b,\ell,\kappa)$-reduction, if the following holds:
	
	\begin{itemize}
		\item For every $x,y \in \{0,1\}^{b \cdot \ell}$,
		\[
		x \cdot y = 0 \Leftrightarrow \varphi(x) \cdot \varphi(y) \in V.
		\]
		\item For every $x$ and $i \in [\ell]$,
		\[
		0 \le \varphi(x)_i < {\ell}^{\kappa \cdot b}.
		\]
	\end{itemize}

	Similarly, let $\tau$ be an increasing function, we say a function family $\{\varphi_{b,\ell}\}_{b,\ell}$ together with a set family $\{ V_{b,\ell} \}_{b,\ell}$ is a $\tau$-reduction family, if for every $b$ and $\ell$, $(\varphi_{b,\ell},V_{b,\ell})$ is a $(b,\ell,\tau(b))$-reduction.

	Moreover, if for all $b$ and all $\ell \le \log\log\log b$, there is an algorithm $\alg$ which computes $\varphi_{b,\ell}(x)$ in $\operatorname*{poly}(b)$ time given $b,\ell$ and $x \in \{0,1\}^{b \cdot \ell}$, and constructs the set $V_{b,\ell}$ in $O\left( \ell^{O(\tau(b) \cdot b)} \cdot \operatorname{poly}(b) \right)$ time given $b$ and $\ell$, then we call $(\varphi_{b,\ell},V_{b,\ell})$ a uniform-$\tau$-reduction family.
\end{defi}

\begin{rem}
	The reason we assume $\ell$ to be small is that in our applications we only care about very small $\ell$, and that greatly simplifies the notation. From Theorem~\ref{theo:main-reduction}, there is a uniform-$\left(6^{\logstar b}\right)$-reduction family, and a better uniform-reduction family implies better hardness for $\IntOV$ and other related problems as well (Lemma~\ref{lm:dim-reduction-OV}, Theorem~\ref{theo:IntOV-to-Max-IP}, Lemma~\ref{lm:Max-IP-to-bichromatic-closest-pair} and Lemma~\ref{lm:Max-IP-to-furtherest-pair}).
\end{rem}

Now we are ready to state our nonuniform to uniform transformation result formally.

\begin{theo}\label{theo:nonuniform-to-uniform}
	Letting $\tau$ be an increasing function such that $\tau(n) = O(\log\log\log n)$ and supposing there is a $\tau$-reduction family, then there is a uniform-$O(\tau)$-reduction family.
\end{theo}
\begin{proofsketch}
	The construction in Theorem~\ref{theo:main-reduction} is recursive, it constructs the reduction $\psi_{b,\ell}$ from a much smaller reduction $\psi_{\bm,\ell}$, where $\bm \le \log b$. In the original construction, it takes $\logstar b$ recursions to make the problem sufficiently small so that a direct construction can be used. Here we only apply the reduction thrice. First let us abstract the following lemma from the proof of Theorem~\ref{theo:main-reduction}.
	
	\begin{lemma}[Implicit in Theorem~\ref{theo:main-reduction}]\label{lm:reduction-once}
		Letting $b,\ell,\bm,\kappa \in \mathbb{N}$ and supposing $\ell^{\kappa \cdot \bm} = b$ and there is a $(\bm,\ell,\kappa)$-reduction $(\varphi,V')$, the following holds:
		
		\begin{itemize}
			\item There is a $(b,\ell,6 \cdot \kappa)$-reduction $(\psi,V)$.
			
			\item Given $(\varphi,V')$, for all $x \in \{0,1\}^{b \cdot \ell}$, $\psi(x)$ can be computed in $\operatorname*{poly}(b \cdot \ell)$, and $V$ can be constructed in $O\left( \ell^{O(\kappa \cdot b)} \cdot \operatorname{poly}(b \cdot \ell) \right)$ time.
		\end{itemize}		
	\end{lemma}

	Now, let $b,\ell \in \mathbb{N}$, we are going to construct our reduction as follows.
	
	Let $b_1$ be the number such that 
	\[
	\ell^{\tau(b) \cdot 6^2 \cdot b_1} = b,
	\]
	and similarly we set $b_2$ and $b_3$ so that
	\[
	\ell^{\tau(b) \cdot 6 \cdot b_2} = b_1 \quad \text{ and } \quad \ell^{\tau(b) \cdot b_3} = b_2.
	\]
	
	We can calculate from above that $b_3 \le \log\log\log b$. 
	
	From the assumption that there is a $\tau$-reduction, there is a $(b_3,\ell,\tau(b_3))$-reduction $(\varphi_{b_3,\ell},V_{b_3,\ell})$, which is also a $(b_3,\ell,\tau(b))$-reduction, as $\tau$ is increasing. Note that we can assume $\ell \le \log\log\log b$ and $\tau(b) \le \log \log \log b$ from assumption. Now we simply use a brute force algorithm to find $(\varphi_{b_3,\ell},V_{b_3,\ell})$. There are
	\[
	\ell^{\tau(b) \cdot b_3 \cdot \ell \cdot 2^{b_3 \cdot \ell} } = b^{o(1)}
	\]
	possible functions from $\{0,1\}^{b_3 \cdot \ell} \to \{0,\dotsc \ell^{\tau(b_3) \cdot b_3}-1\}^{\ell}$. Given such a function $\varphi$, one can check in $\poly(2^{b_3 \cdot \ell}) = b^{o(1)}$ time that whether one can construct a corresponding set $V$ to obtain our $(b_3,\ell,\tau(b))$-reduction.
	
	Applying Lemma~\ref{lm:reduction-once} thrice, one obtain a $(b,\ell,O(\tau(b)))$-reduction $(\psi,V)$. And since $\varphi_{b_3,\ell}$ can be found in $b^{o(1)}$ time, together with Lemma~\ref{lm:reduction-once}, we obtain a uniform-$\tau$-reduction family.
	
\end{proofsketch}

Finally, we give a direct corollary of Theorem~\ref{theo:nonuniform-to-uniform} that the existence of an $O(1)$-reduction family implies hardness of $\Hopcroft$, $\IntMaxIP$, $\ell_2$-Furthest Pair and Bichromatic $\ell_2$-Closest Pair in $\omega(1)$ dimensions.

\begin{cor}
	If there is an $O(1)$-reduction family, then for every $\eps > 0$, there exists a $c \ge 1$ such that $\Hopcroft$, $\IntMaxIP$, $\ell_2$-Furthest Pair and Bichromatic $\ell_2$-Closest Pair in $c$ dimensions with $O(\log n)$-bit entries require $n^{2 - \eps}$ time.
\end{cor}
\begin{proofsketch}
	Note that since its hardness implies the harnesses of other three, we only need to consider $\Hopcroft$ here. 
	
	From Theorem~\ref{theo:nonuniform-to-uniform} and the assumption, there exists a uniform-$O(1)$-reduction. Proceeding similar as in Lemma~\ref{lm:dim-reduction-OV} with the uniform-$O(1)$-reduction, we obtain a better dimensionality self reduction from $\OV$ to $\Hopcroft$. Then exactly the same argument as in Theorem~\ref{theo:Hopcroft} with different parameters gives us the lower bound required.
\end{proofsketch}

	\newcommand{\sign}{\mathrm{sgn}}
\newcommand{\WA}{\widetilde{A}}
\newcommand{\WB}{\widetilde{B}}

\newcommand{\peps}{P_\eps}
\newcommand{\hpeps}{\widehat{P}_\eps}
\newcommand{\hcof}{\hat{c}}
\newcommand{\tcof}{\tilde{c}}

\section{Hardness of Approximate $\pnMaxIP$ via Approximate Polynomial for $\OR$}
\label{app:quantum-obs}

We first show that making use of the $O(\sqrt{n})$-degree approximate polynomial for $\OR$~\cite{buhrman1999bounds,de2008note}, $\OV$ can be reduced to approximating $\pnMaxIP$.

\begin{theo}\label{theo:reduction}
	Letting $\eps \in (0,1)$, an $\OV_{n,d}$ instance with sets $A,B$ reduces to a $\pnMaxIP_{n,d_1}$ instance with sets $\WA$ and $\WB$, such that:
    
    \begin{itemize}
	\item $d_1 = \binom{d}{ \le O\left(\sqrt{d \log 1/\eps}\right)}^{3} \cdot 2^{O\left(\sqrt{d\log 1/\eps}\right)} \cdot \eps^{-1}$, in which the notation $\binom{n}{\le m}$ denotes $\sum_{i=0}^{m} \binom{n}{i}$. 
    
    \item There is an integer $T > \eps^{-1}$ such that if there is an $(a,b) \in A \times B$ such that $a \cdot b = 0$, then $\OPT(\WA,\WB) \ge T$.
    
    \item Otherwise, $ |\OPT(\WA,\WB)| \le T \cdot \eps$.
    
    \item Moreover, the reduction takes $n \cdot \poly(d_1)$ time.
    \end{itemize}
\end{theo}

We remark here that the above reduction fails to achieve a characterization: setting $\eps = 1/2$ and $d = c\log n$ for an arbitrary constant $c$, we have $d_1 = 2^{\widetilde{O}(\sqrt{\log n})}$, much larger than $\log n$. Another interesting difference between the above theorem and Lemma~\ref{lm:OV-to-MaxIP} (the reduction from $\OV$ to approximating $\MaxIP$) is that Lemma~\ref{lm:OV-to-MaxIP} reduces one $\OV$ instance to many $\MaxIP$ instances, while the above reduction only reduces it to one $\pnMaxIP$ instance.

\begin{proofof}{Theorem~\ref{theo:reduction}}$\newline$
	\medskip \noindent \textbf{Construction and Analysis of Polynomial $P_{\eps}(z)$.} By~\cite{buhrman1999bounds,de2008note}, there is a polynomial $P_{\eps} : \{0,1\}^d \to \R$ such that:
	
	\begin{itemize}
		\item $P_{\eps}$ is of degree $D = O\left( \sqrt{d \log 1/\eps} \right)$.
		\item For every $z \in \{0,1\}^{d}$, $P_{\eps}(z) \in [0,1]$.
		\item Given $z \in \{0,1\}^{d}$, if $\OR(z) = 0$, then $P_{\eps}(z) \ge 1 - \eps$, otherwise $P_\eps(z) \le \eps$.
		\item $P_{\eps}$ can be constructed in time polynomial in its description size.
	\end{itemize}

	Now, let us analyze $P_\eps$ further. For a set $S \subseteq [d]$, let $\chi_S : \{0,1\}^{d} \to \R$ be $\chi_S(z) := \prod_{i \in S} (-1)^{z_i}$. Then we can write $P_\eps$ as:
	\[
	P_{\eps} := \sum_{S \subseteq [d], |S| \le D} \chi_S \cdot \langle \chi_S,P_\eps \rangle,
	\]
	where $\langle \chi_S,P_\eps \rangle$ is the inner product of $\chi_S$ and $P_\eps$, defined as $\langle \chi_S,P_\eps \rangle := \Ex_{x \in \{0,1\}^{d}} \chi_S(x) \cdot P_\eps(x)$.
	
	Let $c_S = \langle \chi_S,P_\eps \rangle$, from the definition it is easy to see that $c_S \in [-1,1]$. 
	
	\medskip \noindent \textbf{Discretization of Polynomial $P_{\eps}$.} Note that $P_{\eps}(z)$ has real coefficients, we need to turn it into another polynomial with integer coefficients first. 
	
	Let $M := \binom{d}{\le D}$, consider the following polynomial $\hpeps$:
	\[
	\hpeps := \sum_{S \subseteq [d], |S| \le D} \lfloor c_S \cdot 2M/\eps  \rfloor \cdot \chi_S.
	\]
	
	We can see that $| \hpeps(z)/ (2M/\eps) - \peps(z) | \le \eps$ for every $z \in \{0,1\}^d$, and we let $\hcof_S := \lfloor c_S \cdot M \cdot 2/\eps  \rfloor$ for convenience.
	
	\medskip \noindent \textbf{Simplification of Polynomial $\hpeps$.} $\hpeps(z)$ is expressed over the basis $\chi_S$'s, we need to turn it into a polynomial over standard basis.
	
	For each $S \subseteq [d]$, consider $\chi_S$, it can also be written as:
	\[
	\chi_S(z) = \prod_{i \in S} (-1)^{z_i} := \prod_{i \in S} (1-2z_i) = \sum_{ T \subseteq S} (-2)^{|T|} z_T,
	\]
	where $z_T := \prod_{i \in T} z_i$. Plugging it into the expression of $\hpeps$, we have
	\[
	\hpeps(z) := \sum_{T \subseteq [d], |T| \le D} \left( \sum_{S \subseteq [d], |S| \le D, T \subseteq S} \hcof_S \right) \cdot (-2)^{|T|} z_T.
	\]
	
	Set
	\[
	\tcof_T := \left( \sum_{S \subseteq [d], |S| \le D, T \subseteq S} \hcof_S \right) \cdot (-2)^{|T|},
	\]
	the above simplifies to
	\[
	\hpeps(z) := \sum_{T \subseteq [d], |T| \le D} \tcof_T \cdot z_T.
	\]
	
	\medskip \noindent \textbf{Properties of Polynomial $\hpeps$.}
	Let us summarize some properties of $\hpeps$ for now. First we need a bound on $|\tcof_T|$, we can see $|\hcof_S| \le M \cdot 2/\eps$, and by a simple calculation we have
	\[
	|\tcof_T| \le M^2 \cdot 2^{D} \cdot 2/\eps.
	\]
	
	Let $B = M^2 \cdot 2^{D} \cdot 2/\eps$ for convenience. For $x,y \in \{0,1\}^d$, consider $\hpeps(x,y) := \hpeps(x_1y_1,x_2y_2,\dotsc,x_dy_d)$ (that is, plugging in $z_i = x_iy_i$), we have
	\[
	\hpeps(x,y) := \sum_{T \subseteq [d], |T| \le D} \tcof_T \cdot x_T \cdot y_T,
	\]
	where $x_T := \prod_{i \in T} x_i$ and $y_T$ is defined similarly. Moreover, we have
	
	\begin{itemize}
		\item If $x \cdot y = 0$, then $\hpeps(x,y) \ge (2M /\eps) \cdot (1-2\eps) $.
        \item If $x \cdot y \ne 0$, then $|\hpeps(x,y)| \le (2M /\eps) \cdot 2\eps$.
	\end{itemize}
	
	\medskip \noindent \textbf{The Reduction.}	
	Now, let us construct the reduction, we begin with some notations. For two vectors $a,b$, we use $a \circ b$ to denote their concatenation. For a vector $a$ and a real $x$, we use $a \cdot x$ to denote the vector resulting from multiplying each coordinate of $a$ by $x$.  Let $\sign(x)$ be the sign function that outputs $1$ when $x > 0$, $-1$ when $x < 0$, and $0$ when $x = 0$. 
	 For $x \in \{-B,-B+1,\dotsc,B\}$, we use $e_x \in \{-1,0,1\}^{B}$ to denote the vector whose first $|x|$ elements are $\sign(x)$ and the rest are zeros. We also use $\mathbf{1}$ to denote the all-$1$ vector with length $B$. 
	
	Let $T_1,T_2,\dotsc,T_{M}$ be an enumeration of all subsets $T \subseteq [d]$ such that $|T| \le D$, we define
	\[
	\varphi_x(x) := \circ_{i=1}^{M} (e_{\tcof_{T_i}} \cdot x_{T_i}) \text{ and } 
	\varphi_y(y) := \circ_{i=1}^{M} (\mathbf{1} \cdot y_{T_i}).
	\]
	
	And we have
	\[
	\varphi_x(x) \cdot \varphi_y(y) = \sum_{i=1}^{M} (e_{\tcof_{T_i}} \cdot \mathbf{1}) \cdot (x_{T_i} \cdot y_{T_i}) = \sum_{i=1}^{M} \tcof_{T_i} \cdot x_{T_i} \cdot y_{T_i} = \hpeps(x,y).
	\]
	
	To move from $\{-1,0,1\}$ to $\{-1,1\}$, we use the following carefully designed reductions $\psi_x,\psi_y : \{-1,0,1\} \to \{-1,1\}^{2}$, such that 
	\[
	\psi_x(-1) = \psi_y(-1) = (-1,-1),\quad \psi_x(0) = (-1,1), \quad \psi_y(0) := (1,-1),~~\text{ and }~~\psi_x(1) = \psi_y(1) = (1,1).
	\]
	
	It is easy to check that for $x,y \in \{-1,0,1\}$, we have $\psi_x(x) \cdot \psi_y(y) = 2 \cdot (x \cdot y)$.
	
	Hence, composing the above two reductions, we get our desired reductions $\phi_x = \psi_x^{\otimes (B \cdot M)} \circ \varphi_x$ and $\phi_y = \psi_y^{\otimes (B \cdot M)} \circ \varphi_y$ such that for $x,y \in \{0,1\}^d$, $\phi_x(x),\phi_y(y) \in \{-1,1\}^{2B \cdot M}$ and $\phi_x(x) \cdot \phi_y(y) = 2 \cdot \hpeps(x,y)$.
	
	Finally, given an $\OV_{n,d}$ instance with two sets $A$ and $B$, we construct two sets $\WA$ and $\WB$, such that $\WA$ consists of all $\phi_x(x)$'s for $x \in A$, and $\WB$ consists of all $\phi_y(y)$'s for $y \in B$.
	
	Then we can see $\WA$ and $\WB$ consist of $n$ vectors from $\{-1,1\}^{d_1}$, where 
	\[
	d_1 = 2B \cdot M = M^3 \cdot 2^{D} \cdot 2/\eps = \binom{d}{ \le O\left(\sqrt{d \log 1/\eps}\right)}^{3} \cdot 2^{O\left(\sqrt{d\log 1/\eps}\right)} \cdot \eps^{-1}
	\] as stated.
	
	It is not hard to see the above reduction takes $n \cdot \poly(d_1)$ time. Moreover, if there is a $(x,y) \in A \times B $ such that $x \cdot y = 0$, then $\OPT(\WA,\WB) \ge (4M /\eps) \cdot (1-2\eps)$, otherwise, $\OPT(\WA,\WB) \le (4M/\eps) \cdot 2\eps$. Setting $\eps$ above to be $1/3$ times the $\eps$ in the statement finishes the proof.
\end{proofof}

With Theorem~\ref{theo:reduction}, we are ready to prove our hardness results on $\pnMaxIP$.

\begin{theo}\label{theo:quantum-based-lowb}
	Assume SETH (or OVC). Letting $\alpha : \mathbb{N} \to \R$ be any function of $n$ such that $\alpha(n) = n^{o(1)}$, there is another function $\beta$ satisfying $\beta(n) = n^{o(1)}$ and an integer $T > \alpha$ ($\beta$ and $T$ depend on $\alpha$), such that there is no $n^{2-\Omega(1)}$-time algorithm for $\pnMaxIP_{n,\beta(n)}$ distinguishing the following two cases:
    \begin{itemize}
    	\item $\OPT(A,B) \ge T$ ($A$ and $B$ are the sets in the $\pnMaxIP$ instance).
        \item $|\OPT(A,B)| \le T / \alpha(n)$.
    \end{itemize}
\end{theo}
\begin{proof}
	
	Letting $\alpha = n^{o(1)}$ and $ k = \log \alpha / \log n$, we have $k = o(1)$. Setting $d = c \log n$ where $c$ is an arbitrary constant and $\eps = \alpha^{-1}$ in Theorem~\ref{theo:reduction}, we have that an $\OV_{c \log n}$ reduces to a certain $\alpha(n)$-approximation to a $\pnMaxIP_{n,d_1}$ instance with sets $A$ and $B$, where 
	$$
	d_1 = \binom{c\log n}{\le O(\sqrt{ck} \log n)}^3 \cdot 2^{O(\sqrt{ck} \log n)} \le \left( \frac{\sqrt{c}}{\sqrt{k}} \right)^{O(\sqrt{ck} \log n)} \cdot 2^{O(\sqrt{ck} \log n)} = n^{O( \log (c/k) \cdot \sqrt{ck})}.
	$$
	Now set $\beta = n^{k^{1/3}}$ and $T$ be the integer specified by Theorem~\ref{theo:reduction}, since $k = o(1)$, $\beta = n^{o(1)}$. Suppose otherwise there is an $n^{2-\Omega(1)}$-time algorithm for distinguishing whether $\OPT(A,B) \ge T$ or $|\OPT(A,B)| \le T / \alpha(n)$. Then for any constant $c$, $O(\log(c/k) \sqrt{ck}) \le k^{1/3}$ for sufficiently large $n$, which means $d_1 \le \beta(n)$ for a sufficiently large $n$, and there is an $n^{2-\Omega(1)}$-time algorithm for $\OV_{c \log n}$ by Theorem~\ref{theo:reduction}, contradiction to OVC.
\end{proof}

	\section{A Proof of Lemma~\ref{lm:OV-to-MaxIP}}
\label{app:OV-MaxIP-reduction}

Finally, we present a proof of Lemma~\ref{lm:OV-to-MaxIP}, which is implicit in~\cite{Rubinstein2017closest}.

We need the following efficient $\MA$ protocol for Set-Disjointness from~\cite{Rubinstein2017closest}, which is also used in~\cite{karthik2017parameterized}.\footnote{The protocol in~\cite{karthik2017parameterized} also works for the $k$-party number-in-hand model.}

\begin{lemma}[Theorem~3.2 of~\cite{Rubinstein2017closest}]\label{lm:MA-protocol-prev}
	For every $\alpha$ and $m$, there is an $(m/\alpha,\log_2 m,\poly(\alpha),1/2)$-efficient $\MA$ protocol for $\DISJ_{m}$.
\end{lemma}

We want to reduce the error probability while keeping the number of total random coins relatively low. To achieves this, we can use an expander graph (Theorem~\ref{theo:derand}) to prove the following theorem.

\begin{lemma}\label{lm:MA-protocol}
	For every $\alpha$, $m$ and $\eps < 1/2$, there is an $(m/\alpha,\log_2 m + O(\log \eps^{-1}),\poly(\alpha) \cdot \log \eps^{-1} ,\eps)$-efficient $\MA$ protocol for $\DISJ_{m}$.
\end{lemma}

\begin{proof}
	Let $c_1$ and $\mathcal{F} : \{0,1\}^{\log m + c_1 \cdot \log \eps^{-1}} \to [m]^{c_1 \cdot \log \eps^{-1}}$ be the corresponding constant and function as in Theorem~\ref{theo:derand}, and let $\Pi$ denote the $(m/\alpha,\log_2 m,\poly(\alpha),1/2)$-efficient $\MA$ protocol for $\DISJ_{m}$ in Lemma~\ref{lm:MA-protocol-prev}. Set $q = c_1 \cdot \log \eps^{-1}$ and our new protocol $\Pi_{\textsf{new}}$ works as follows:
	
	\begin{itemize}
		\item Merlin still sends the same advice to Alice as in $\Pi$.
		
		\item Alice and Bob jointly toss $r = \log m + q$ coins to get a string $w \in \{0,1\}^{r}$. Then we let $w_1,w_2,\dotsc,w_{q}$ be the sequence corresponding to $\mathcal{F}(w)$, each of them can be interpreted as $\log m$ bits.
		
		\item Bob sends Alice $q$ messages, the $i$-th message $m_i$ corresponds to Bob's message in $\Pi$ when the random bits is $w_i$.
		
		\item After that, Alice decides whether to accept or not as follows:
		
		\begin{itemize}
			\item If for every $i \in [q]$, Alice would accept Bob's message $m_i$ with random bits $w_i$ in $\Pi$, then Alice accepts.
			
			\item Otherwise, Alice rejects.
		\end{itemize}
	\end{itemize}
	
	It is easy to verify that the advice length, message length and number of random coins satisfy our requirements. 
	
	For the error probability, note that when these two sets are disjoint, the same advice in $\Pi$ leads to acceptance of Alice. Otherwise, suppose the advice from Merlin is either wrong or these two sets are intersecting, then half of the random bits in $\{0,1\}^{\log m}$ leads to the rejection of Alice in $\Pi$. Hence, from Theorem~\ref{theo:derand}, with probability at least $1 - \eps$, at least one of the random bits $w_i$'s would lead to the rejection of Alice, which completes the proof.
\end{proof}

Finally we are going to prove Lemma~\ref{lm:OV-to-MaxIP}, we recap it here for convenience.

\begin{reminder}{Lemma~\ref{lm:OV-to-MaxIP}} 
	There is a universal constant $c_1$ such that, for every integer $c$, reals $\varepsilon \in (0,1]$ and $\tau \ge 2$, $\OV_{n,c \log n}$ can be reduced to $n^{\varepsilon}$ $\MaxIP_{n,d}$ instances $(A_i,B_i)$ for $i \in [n^{\varepsilon}]$, such that:
	
	\begin{itemize}
		\item $d = \tau^{\poly(c /\varepsilon) } \cdot \log n$.
		\item Letting $T = c \log n \cdot \tau^{c_1}$, if there is $a \in A$ and $b \in B$ such that $a \cdot b = 0$, then there exists an $i$ such that $\OPT(A_i,B_i) \ge T$.
		\item Otherwise, for all $i$ we must have $\OPT(A_i,B_i) \le T/\tau$.
	\end{itemize}
\end{reminder}
\begin{proof}
	The reduction follows exactly the same as in~\cite{ARW17-proceedings}, we recap here for completeness.
	
	Set $\alpha = c / \eps$, $m = c \cdot \log n$ and $\eps = 1/\tau$, and let $\Pi$ be the $(m/\alpha,\log_2 m + O(\log \eps^{-1}),\poly(\alpha) \cdot \log \eps^{-1} ,\eps)$-efficient $\MA$ protocol for Set-Disjointness as in Lemma~\ref{lm:MA-protocol}.
	
	Now, we first enumerate all of $2^{m/\alpha} = 2^{\eps \cdot \log n} = n^{\varepsilon}$ possible advice strings, and create an $\MaxIP$ instance for each of the advice strings.
	
	For a fix advice $\psi \in \{0,1\}^{\eps \cdot \log n}$, we create an $\MaxIP$ instance with sets $A_\psi$ and $B_\psi$ as follows. We use $a \circ b$ to denote the concatenation of the strings $a$ and $b$.
	
	Let $r = \log_2 m + c_1 \cdot \log \eps^{-1}$, where $c_1$ is the constant hidden in the big $O$ notation in Lemma~\ref{lm:MA-protocol}, and $\ell = \poly(\alpha) \cdot \log \eps^{-1}$. Let $m_1,m_2,\dotsc,m_{2^\ell}$ be an enumeration of all strings in $\{0,1\}^{\ell}$.
	
	\begin{itemize}
		\item For each $a \in A$, and for each string $w \in \{0,1\}^{r}$, we create a vector $a^w \in \{0,1\}^{2^\ell}$, such that $a^{w}_i$ indicates that given advice $\psi$ and randomness $w$, whether Alice accepts message $m_i$ or not ($1$ for acceptance, $0$ for rejection). Let the concatenation of all these $a^w$'s be $a_\psi$. Then $A_\psi$ is the set of all these $a_\psi$'s for $a \in A$.
		
		\item For each $b \in B$, and for each string $w \in \{0,1\}^r$, we create a vector $b^{w} \in \{0,1\}^{2^\ell}$, such that $b^{w}_i = 1$ if Bob sends the message $m_i$ given advice $\psi$ and randomness $w$, and $=0$ otherwise. Let the concatenation of all these $b^w$'s be $b_\psi$. Then $B_\psi$ is the set of all these $b_\psi$'s for $b \in B$.
	\end{itemize}
	
	We can see that for $a \in A$ and $b \in B$, $a_\psi \cdot b_\psi$ is precisely the number of random coins leading Alice to accept the message from Bob given advice $\psi$ when Alice and Bob holds $a$ and $b$ correspondingly. Therefore, let $T = 2^{r} = c \log n \cdot \tau^{c_1}$, from the properties of the protocol $\Pi$, we can see that:
	
	\begin{itemize}
		\item If there is $a \in A$ and $b \in B$ such that $a \cdot b = 0$, then there is $\psi \in \{0,1\}^{\eps \cdot \log n}$ such that $a_\psi \cdot b_\psi \ge T$.
		
		\item Otherwise, for all $a \in A$, $b \in B$ and advice $\psi \{0,1\}^{\eps \cdot \log n}$, $a_\psi \cdot b_\psi \le T / \tau$.
	\end{itemize}
	
	And this completes the proof.
	
\end{proof}
\bibliographystyle{alpha}
\bibliography{team}
	
\end{document}